\documentclass[11pt]{amsart}
\usepackage[utf8]{inputenc}
\usepackage[margin=1in]{geometry}

\usepackage{upgreek}
\usepackage{amssymb, amsmath,amsthm}
\usepackage{amsfonts}
\usepackage{float}
\usepackage{graphics}
\usepackage{subfigure}
\usepackage{color}
\usepackage{graphicx}
\usepackage{enumitem}
\usepackage{lineno}
\usepackage{pgfplots}
\pgfplotsset{compat=1.12}
\usepackage{pgfplotstable}

\usepackage[colorlinks=true,citecolor=blue]{hyperref}
\hypersetup{linkcolor=red}

\newcommand{\abs}[1]{\lvert #1 \rvert}
\newcommand{\Fr}{\text{Fr}}
\newcommand{\Ro}{\text{Ro}}
\newcommand{\norm}[1]{\| #1 \|}
\newcommand{\ip}[1]{\langle #1 \rangle}
\def\bfu{{\bf u}}

\renewcommand{\Re}{\operatorname{Re}}
\renewcommand{\Im}{\operatorname{Im}}

\newtheorem{lemma}{Lemma}[section]
\newtheorem{theorem}{Theorem}[section]

\newtheorem*{remark}{Remark}

\numberwithin{equation}{section}
\numberwithin{figure}{section}

\begin{document}
\title[Dynamic Transitions and Baroclinic Instability]{Dynamic Transitions and Baroclinic Instability for 3D Continuously Stratified Boussinesq Flows}

\author[Şengül]{Taylan Şengül}
\address[Şengül]{Department of Mathematics, Marmara University, 34722 Istanbul, Turkey}
\email{taylan.sengul@marmara.edu.tr}

\author[Wang]{Shouhong Wang}
\address[Wang]{Department of Mathematics,
Indiana University, Bloomington, IN 47405}
\email{showang@indiana.edu, http://www.indiana.edu/\~{ }fluid}
\thanks{This research is supported in part by the National Science Foundation (NSF) grant DMS-1515024, and by the Office of Naval Research (ONR) grant N00014-15-1-2662.}

\begin{abstract}
The main objective of this article is to study the nonlinear stability and dynamic transitions of the basic (zonal) shear flows for the three-dimensional continuously stratified rotating Boussinesq model. The model equations are fundamental equations in geophysical fluid dynamics, and dynamics associated with their basic zonal shear flows play a crucial role in understanding many important geophysical fluid dynamical processes, such as the meridional overturning oceanic circulation and the geophysical baroclinic instability.
In this paper, first we derive a threshold for the energy stability of the basic shear flow, and obtain a criterion for local nonlinear stability in terms of the critical horizontal wavenumbers and the system parameters such as the Froude number, the Rossby number, the Prandtl number and the strength of the shear flow.
Next, we demonstrate that the system always undergoes a dynamic transition from the basic shear flow to either a spatiotemporal oscillatory pattern or circle of steady states,
as the shear strength of the basic flow crosses a critical threshold. Also, we show that the dynamic transition can be either continuous or catastrophic, and is dictated by the sign of a transition number, fully characterizing the nonlinear interactions of different modes.  Both the critical shear strength and the transition number are functions of the system parameters. A systematic numerical method is carried out to explore transition in different flow parameter regimes. In particular, our numerical investigations show the existence of a hypersurface which separates the parameter space into regions where the basic shear flow is stable and unstable. Numerical investigations also yield that the selection of horizontal wave indices is determined only by the aspect ratio of the box. We find that the system admits only critical eigenmodes with roll patterns aligned with the x-axis. Furthermore, numerically we encountered continuous transitions to multiple steady states, as well as continuous and catastrophic transitions to spatiotemporal oscillations.
\end{abstract}
\date{May 9, 2017}
\maketitle

\section{Introduction and Main Results}
The atmosphere and the oceans are extremely rich in their organization and complexity, and many phenomena that they exhibit, involving a broad range of temporal and spatial scales \cite{charney48}, cannot be reproduced in the laboratory.
The behavior of the atmosphere, the ocean, or the coupled atmosphere and ocean can be viewed as an initial and boundary value problem \cite{bjerknes, rossby26, phillips}.
According to \cite{neumann}, depending on the time scale of the phenomenon under consideration, or equivalently, on how far into the future one is attempting to make predictions of atmospheric conditions, it is convenient to consider motions corresponding to short-, medium-, and long-term behavior of the atmosphere.

The ideas of dynamical systems theory were pioneered in \cite{La, Lb,S,Va, Vb, GC, DijkstraB2000, DG05},
among many others, whose authors explored the dynamics of low-order
models of atmospheric and oceanic flows. The formulation of the primitive equations of atmosphere and ocean as an infinite-dimensional dissipative dynamical system were initiated in \cite{LTW92a}.
It is evident that geophysical fluid flows and climate variability exhibit recurrent large-scale patterns which are directly linked to dynamical processes represented in the governing dissipative dynamical system of the atmosphere and the ocean. The study of the persistence of these patterns and the transitions between them are of fundamental importance in geophysical fluid dynamics and climate dynamics.

The main objective of this paper is to study the stability and dynamic transitions of a basic shear flow, associated with geophysical baroclinic instability, for the three dimensional (3D) continuously stratified rotating Boussinesq model. Baroclinic instability is one of the most important geophysical fluid dynamical instability and plays a crucial role in understanding the dominant mechanism shaping the cyclones and anticyclones that dominate the weather in mid-latitudes, as well as the mesoscale eddies that play various roles in oceanic dynamics and the transport of tracers.

The first analytic study of the problem was made by Eric Eady \cite{eady49}, Joseph Pedlosky \cite{ped1970} in the context of a single unstable baroclinic wave, followed by the work of Mak \cite{mak1985}, Cai and Mak \cite{cai87} and Cai \cite{cai92}, among many others. In view of all the previous work for this problem and to the best of our knowledge, the study of nonlinear stability and dynamic transitions of the basic shear flows, however, is still open.

Also, the previous studies on baroclinic instability are based on simplified models such as the
3D quasi-geostrophic model and multi-layer (in vertical) models. The model we adopt in this article is the 3D rotating continuously stratified Boussinesq equations, which are fundamental equations in geophysical fluid dynamics. They describe the conservation of mass, the conservation of momentum, coupled with the first law of thermodynamics, under the Boussinesq assumption that the density is constant except in the buoyancy term and in the equation of state.

The basic shear flow, in its nondimensional form, is given by \eqref{basic flow}:
\[
  \bfu_{ss} = (\Lambda z, 0, 0), \quad
  p_{ss} = - \frac{\Lambda}{\Ro} yz \quad
  \rho_{ss} = \frac{\Fr^2 \Lambda}{\Ro} y,
\]
where $\bfu_{ss}$ is the vertically dependent zonal shear velocity, $p_{ss}$ and $\rho_{ss}$ are the corresponding pressure and density fields. Here $\Fr$ is the Froude number, $\Lambda$ is the (non-dimensional) strength of the basic shear flow, $\Pr$ is the Prandtl number and $\Ro$ is the Rossby number. For simplicity of the analysis and presentation, we use periodicity in the horizontal directions coupled with no-slip boundary conditions in the vertical direction for the deviation fields.

\medskip

We now describe the main results of this paper.

First, we find a threshold for the energy stability of the system. We prove that if
$$\Lambda < \Fr^4 \Ro \Pr \pi^4,$$
then for any initial data, the energy of the deviation field decays in the $L^2$-norm at least exponentially in time. This result suggests that increasing any of the parameters $\Fr$, $\Ro$, $\Pr$ has a stabilizing effect whereas increasing $\Lambda$ has a destabilizing effect on the basic shear flow.

Next we turn our attention to the linear stability problem. Due to the periodicity of the system in the horizontal directions, the determination of the eigenpairs of the linear operator is reduced to the study of the eigenvalue problem corresponding to the z-direction of the eigenvectors. The resulting equations for non-zero horizontal wave numbers $\alpha$ have to be solved numerically. We use a Legendre-Galerkin method to solve these equations which provides a very efficient method for the determination of the eigenpairs.

Studying the eigenvalue problem analytically, we show that if the wavenumber $\alpha_c$ corresponding to the eigenvalue with largest real part satisfies
\[
  \alpha_c^2 > \left( \frac{1}{\Fr^4} + \frac{\sqrt{2}}{\Ro} \right) \max \left\{\left(\frac{\Lambda^2}{\Pr}\right)^{1/3}, \frac{1}{\Pr} \right\}
\]
then the system is locally nonlinearly stable.

Our numerical experiments with the linearized eigenvalue problem show the existence of a critical parameter $\Lambda_c$ for which the eigenvalue $\beta_1$ with the largest real part, the critical eigenvalue, crosses the imaginary axis at $\Lambda = \Lambda_c$. Since we are not able to find an analytical expression for $\Lambda_c$ in terms of other system parameters, we are led to numerical explorations of this critical parameter.

One of the main issues in this paper is the identification of the first transition which takes place at $\Lambda = \Lambda_c$.
The first transition of the basic shear flow depends on the dimension of the critical eigenspace. Due to the periodicity condition in the horizontal directions, the first critical eigenspace spanned by the critical eigenmodes is always even dimensional, generically two. Thus we study the critical crossing of double real eigenvalues and a pair of complex conjugate eigenvalues separately.

The critical crossing of double real eigenvalues lead to a first transition of the basic shear flow to multiple steady states while the critical crossing of complex pair of eigenvalues lead to a first transition to a spatial-temporal oscillations (time periodic solutions). In both cases, as the shear strength $\Lambda$ crosses the critical value $\Lambda_c$, the system undergoes a continuous transition or catastrophic transition from the basic shear flow, dictated by a common parameter $A$, which characterizes precisely the interactions between different modes of the system.

In the case of transition to multiple steady states, we show the existence of a circle of degenerate steady states bifurcating at $\Lambda=\Lambda_c$.

We also prove a new representation for the approximation of the center manifold function for the critical crossing of complex conjugate pair of eigenvalues which is suitable for studying problems which are extended periodically in at least one spatial direction. This representation leads to the determination of a transition number which describes the critical crossing of both double real eigenvalues and a pair of complex conjugate eigenvalues together.

We show that the type of transition is determined by a parameter $A$ which characterizes the nonlinear interactions between critical modes with modes having zero wave number and modes having twice the critical wave number. In our numerical investigations we encounter both continuous and catastrophic transitions to
spatiotemporal oscillations and continuous transitions to multiple steady states.

With regards to the pattern formation after the transition, our numerical experiments suggest that the system always prefers rolls parallel to the east-west direction regardless of the east west length scale ratio.

This paper arises out of a research program to generate rigorous mathematical results on geophysical fluid dynamics and on climate variability developed from the viewpoint of dynamical transitions. This type of physically-induced stability and transition leads naturally for us to search for the full set of transition states, represented by a local attractor, giving a complete characterization of stability and transition. Such study was initiated in early 2000 by Ma and Wang, and the corresponding dynamical transition theory and its various applications are synthesized in \cite{ptd}; we have shown in particular that the transitions of all dissipative systems can be classified into three classes: continuous, catastrophic and random, which correspond to very different dynamical transition behavior of the system.
We refer interested readers to \cite{ptd, b-book, MW09c} and the references therein for a wide range of applications of the theory, as well as the recent work \cite{sengul-qgflow,Ozer201671} for dynamic transition of shear flows.

The paper is organized as follows. Section~2 introduces the model and its mathematical setting. Section~3 proves a nonlinear energy stability of the basic shear flow, and Section 4 characterizes linear instability and the principle of exchange of stabilities (PES). A local nonlinear stability result is proved in Section~5 under a condition on the horizontal critical wave numbers. The main dynamic transition theorem, \autoref{main thm}, of the paper is stated in Section~6 and proved in Section~7. Sections~8 and 9 are devoted to numerical investigations of the dynamic transitions of the basic shear flow on various parameter regimes, and Section 10 provides a brief conclusion of the main results in this paper.

\section{Model Equations}
The dimensional equations describing the baroclinicity for continuously stratified Boussinesq flow is given by
\begin{equation} \label{dim main equs}
  \begin{aligned}
    & \bfu_t + (\bfu \cdot \nabla) \bfu - \nu \Delta \bfu + \omega \hat{e}_3 \times \bfu + \frac{1}{\rho_0} \nabla p + \frac{1}{\rho_0} \rho g \hat{e}_3 = 0,\\
    & \rho_t + (\bfu \cdot \nabla) \rho - \kappa \Delta \rho = 0, \\
    & \nabla \cdot \bfu = 0.
  \end{aligned}
\end{equation}
Here $\bfu$ is the velocity field, $p$ is the pressure, $\rho$ is the density (temperature), $\nu$ is the kinematic viscosity, $\omega$ is the planetary rotation rate, $\hat{e}_3$ is the unit vector in the $z$ direction, $\rho_0$ is a reference density, $g$ is the gravitation constant, and $\kappa$ is the thermal diffusivity. It is assumed that the fluid resides in a container with dimensional height $h$.

We nondimensionalize the equations \eqref{dim main equs} using
\[
  x = h x', \quad t = \frac{h^2}{\kappa} t', \quad
  \bfu = \frac{\kappa}{h} \bfu', \quad \rho = \rho_0 \rho',
\]
and drop primes to obtain
\begin{equation} \label{pre-main equs}
  \begin{aligned}
    & \bfu_t + (\bfu \cdot \nabla) \bfu - \Pr \Delta \bfu + \frac{1}{\Ro}\hat{e}_3 \times \bfu + \nabla p + \frac{1}{\Fr^2}\rho \hat{e}_3 = 0, \\
    & \rho_t + (\bfu \cdot \nabla) \rho - \Delta\rho = 0, \\
    & \nabla \cdot \bfu = 0,
  \end{aligned}
\end{equation}
where the nondimensional numbers are defined as
\[
  \Pr = \frac{\nu}{\kappa}, \quad
  \Ro = \frac{[U]}{[L]\omega} = \frac{\kappa}{h^2 \omega}, \quad
  \Fr^2 = \frac{[U]^2}{[L]g} = \frac{\kappa^2}{g h^3}.
\]

The equations \eqref{pre-main equs} admit the following steady state solution characterizing a shearing motion which, due to the Coriolis forces, must be balanced by a pressure gradient in response to spatially varying density field:
\begin{equation} \label{basic flow}
  \bfu_{ss} = (\Lambda z, 0, 0), \quad
  p_{ss} = - \frac{\Lambda}{\Ro} yz, \quad
  \rho_{ss} = \frac{\Fr^2 \Lambda}{\Ro} y.
\end{equation}

Taking the deviations from the basic state
\[
  \bfu' = \bfu - \bfu_{ss}, \quad
  p' = p - p_{ss}, \quad
  \rho' = \rho - \rho_{ss},
\]
and dropping the primes again, we deduce the following equations for the deviation fields:
\begin{equation} \label{main equs}
  \begin{aligned}
    & \bfu_t + (\bfu \cdot \nabla) \bfu + \Lambda z \frac{\partial \bfu}{\partial x} + \Lambda w \hat{e}_3 - \Pr \Delta \bfu + \frac{1}{\Ro}\hat{e}_3 \times \bfu + \nabla p + \frac{1}{\Fr^2}\rho \hat{e}_3 = 0, \\
    & \rho_t + (\bfu \cdot \nabla) \rho + \frac{\Lambda}{\Ro \Fr^2}v - \Delta\rho = 0,\\
    & \nabla \cdot \bfu = 0,
  \end{aligned}
\end{equation}
where $\bfu = (u, v , w)$.
The nondimensional spatial rectangular domain is
\[
  (x, y, z) \in \Omega = (0,L_x) \times (0, L_y) \times (0, 1).
\]

In this study, we consider the periodic boundary conditions for the perturbed variables $(\bfu, p, \rho)$ in \eqref{main equs} in the $x$ and $y$ directions with periods $L_x$ and $L_y$ for simplicity. In the z direction we consider no-slip boundary conditions:
\begin{equation} \label{BC}
  u = v = w = \rho = 0, \qquad \text{at } z=0,1.
\end{equation}
We also want to give a remark for the case of free-slip boundary conditions, i.e. the Neumann boundary conditions for the horizontal velocities in the vertical direction. In the case of free-slip boundaries, the linear problem will always have a pair of eigenvalues with zero real part. In this case the reduction to the center manifold can be carried out by restricting the analysis to space of functions with zero mean.

The functional setting of the problem can be cast as follows. Let
\[
  X_1 =\{ (\bfu, \rho) \in H^2(\Omega, \mathbb{R}^4) : \nabla \cdot \bfu = 0, \, \bfu\mid_{z=0,1} =0, \, \rho\mid_{z=0,1} = 0 \}
\]
and
\[
  X =\{ (\bfu, \rho) \in L^2(\Omega, \mathbb{R}^4) : \nabla \cdot \bfu = 0, \, \bfu\mid_{z=0,1} =0, \, \rho\mid_{z=0,1} = 0 \}
\]
where the spaces denote the usual Sobolev spaces such that $\bfu$, $\rho$ are periodic in $x$ and $y$ with periods $L_x$ and $L_y$.

Now the main equations can be cast as
\begin{equation} \label{abstract equ}
  \frac{d \phi}{dt} = L \phi + G(\phi), \qquad {\phi=(\bfu, \rho)},
\end{equation}
where $L:X_1 \to X$ is the linear operator and $G:X_1\times X_1 \to X$ is the bilinear operator, defined by
\begin{equation}\label{operators}
\begin{aligned}
&  L\phi = -{\mathcal P} \left(\begin{matrix}
      \Lambda z \frac{\partial \bfu}{\partial x} + \Lambda w \hat{e}_3 - \Pr \Delta \bfu + \frac{1}{\Ro}\hat{e}_3 \times \bfu  + \frac{1}{\Fr^2}\rho \hat{e}_3 \\
     \frac{\Lambda}{\Ro \Fr^2}v - \Delta\rho
     \end{matrix}
     \right), \\
  &    G(\phi, \tilde \phi) = -{\mathcal P} \left(\begin{matrix}
  (\bfu \cdot \nabla) \tilde{\bfu} \\
  (\bfu \cdot \nabla) \tilde{\rho}
  \end{matrix}
  \right), \\
  & G(\phi)= G(\phi, \phi).
\end{aligned}
\end{equation}
Here $\mathcal P: L^2(\Omega, \mathbb R^4) \to X$ is the standard Leray projection.

\section{Energy Stability}
In this section, we address the energy stability of the basic shear flow.
\begin{theorem} \label{Energy stability theorem}
  If
  \begin{equation} \label{energy-stability condition}
    \Lambda < \text{\rm Fr}^4 \text{\rm Ro}  \Pr \pi^4,
  \end{equation}
  then the energy
  \begin{equation} \label{energy}
    E = \norm{\bfu}^2 + \frac{\text{\rm Ro}}{\Lambda} \norm{\rho}^2
  \end{equation}
  decays in the $L^2(\Omega)$ norm $\norm{\cdot}$ at least exponentially in time.
\end{theorem}

\begin{proof}
Taking $L^2$ inner product between the first two equations of \eqref{main equs} and ($\bfu, \rho$), we obtain
\begin{align}
&  \label{energy-equ1}
 \frac{d \norm{\bfu}^2}{2 dt} + \Lambda \norm{w}^2 + \Pr \norm{\nabla \bfu}^2 + \frac{1}{\Fr^2} \int_{\Omega} \rho w dV= 0, \\
 &  \label{energy-equ2}
  \frac{d \norm{\rho}^2}{2 dt} + \frac{\Lambda}{\Ro \Fr^2} \int_{\Omega} v \rho dV + \norm{\nabla \rho}^2 = 0.
\end{align}
We derive from \eqref{energy}, \eqref{energy-equ1} and \eqref{energy-equ2} that
\[
  \frac{d E}{2 dt} = - \Lambda \norm{w}^2 - \Pr \norm{\nabla \bfu}^2 - \frac{1}{\Fr^2} \int_{\Omega} \rho ( w + v )dV - \frac{\Ro}{\Lambda}\norm{\nabla \rho}^2.
\]

Using the Poincare inequalities (thanks to the no-slip boundary conditions):
$$
\norm{\nabla \bfu}^2 \ge \pi^2 \norm{\bfu}^2, \qquad \norm{\nabla \rho}^2
 \ge \pi^2 \norm{\rho}^2,
$$
we obtain that
\begin{equation} \label{3.4}
\frac{d E}{2 dt} \le - \Lambda \norm{w}^2 - \Pr \pi^2 \norm{\bfu}^2 - \frac{1}{\Fr^2} \int_{\Omega} \rho ( w + v ) dV - \frac{\Ro \pi^2}{\Lambda}\norm{\rho}^2.
\end{equation}

By the Cauchy-Schwarz and the Young inequalities, we have
\[
  \left\vert \int_{\Omega} \rho ( w + v )dV \right\vert \le
  \epsilon \norm{\rho}^2 + \frac{1}{2\epsilon} \norm{w}^2 + \frac{1}{2\epsilon} \norm{v}^2 \le
  \epsilon \norm{\rho}^2 + \frac{1}{2\epsilon} \norm{\bfu}^2.
\]
Consequently, dropping the $-\Lambda \norm{w}^2$ term, we infer from \eqref{3.4} that
\[
  \frac{d E}{2 dt} \le (- \Pr \pi^2 + \frac{1}{2\Fr^2 \epsilon} ) \norm{\bfu}^2 + (-\frac{\Ro \pi^2}{\Lambda} + \frac{\epsilon}{\Fr^2})\norm{\rho}^2.
\]
Choosing $\epsilon = \frac{\Fr^2 \Ro \pi^2}{2\Lambda}$, we arrive at
\[
  \frac{d E}{2 dt} \le \left(- \Pr \pi^2 + \frac{\Lambda}{\Fr^4 \Ro \pi^2} \right) \norm{\bfu}^2 -\frac{\Ro \pi^2}{2\Lambda} \norm{\rho}^2.
\]
By \eqref{energy-stability condition}, $\Pr \pi^2 - \frac{\Lambda}{\Fr^4 \Ro \pi^2} > 0$. Letting
\[
  c_0 = 2\min\left\{\Pr \pi^2 - \frac{\Lambda}{\Fr^4 \Ro \pi^2}, \frac{\pi^2}{2}\right\},
\]
we finally obtain that
\[
  \frac{dE}{dt} \le -c_0 E(t),
\]
and the theorem follows.
\end{proof}

\begin{remark}
\autoref{Energy stability theorem} is still valid if one assumes Dirichlet or Neumann type boundary conditions in the horizontal directions for the deviation fields rather than the periodicity condition assumed in this study.
However the proof is not valid if one considers free-slip boundary conditions in the vertical directions.
\end{remark}

\section{Spectrum of Linearized Problem and Principle of Exchange of Stabilities}

\subsection{Linearized eigenvalue problem}
In this section we consider the eigenvalue problem associated with the linear part of the main equations \eqref{main equs}:
\begin{equation} \label{pde eig prob}
  \begin{aligned}
    & \Pr \Delta \bfu -\Lambda z \frac{\partial \bfu}{\partial x} - \Lambda w \hat{e}_3  - \frac{1}{\Ro}\hat{e}_3 \times \bfu - \nabla p - \frac{1}{\Fr^2}\rho \hat{e}_3 = \beta \bfu, \\
    & -\frac{\Lambda}{\Ro \Fr^2} v + \Delta\rho = \beta \rho , \\
    & \nabla \cdot \bfu = 0,
  \end{aligned}
\end{equation}
together with the boundary conditions \eqref{BC}.

The above eigenvalue system \eqref{pde eig prob} and the boundary conditions \eqref{BC} have eigenvectors of the form
\begin{equation} \label{ansatz}
  \phi_{m_x, m_y, m_z} =
  \begin{bmatrix}
    \bfu \\ \rho
  \end{bmatrix}
  = e^{i \left( \frac{m_x \pi x}{L_x} + \frac{m_y \pi y}{L_y} \right)}
  \begin{bmatrix}
    U_{m_x, m_y, m_z}(z) \\ V_{m_x, m_y, m_z}(z) \\ W_{m_x, m_y, m_z}(z) \\ R_{m_x, m_y, m_z}(z)
  \end{bmatrix},
\end{equation}
with corresponding eigenvalue $\beta_{m_x, m_y, m_z}$
The horizontal wave number is defined by
\begin{equation}\label{h-wave number}
\alpha^2 = \alpha_x^2 + \alpha_y^2, \qquad
(\alpha_x, \alpha_y) = \left(\frac{m_x \pi}{L_x}, \frac{m_y \pi}{L_y}\right).
\end{equation}

{\it Case $\alpha = 0$:} In this case, the corresponding eigenfunctions do not depend on $x$ and $y$ and $\nabla \cdot \bfu = 0$ implies $\frac{\partial w}{\partial z} = 0$, hence $w = 0$ due to the boundary conditions. Thus system \eqref{pde eig prob} reduces to
\[
  \begin{aligned}
  &
  \Pr D^2 U + \frac{1}{\Ro} V = \beta U, \\
&  \Pr D^2 V - \frac{1}{\Ro} U = \beta V, \\
 & D^2 R - \frac{\Lambda}{\Ro \Fr^2} V = \beta R, \\
  & U = V = R = 0 && \text{at } \quad z=0,1,
  \end{aligned}
\]
where $D = \frac{d}{dz}$.
The solutions can be found for $m_z \in \mathbb{Z}^+$, and are given by
\begin{equation} \label{alpha=0 eigenpairs}
  \begin{aligned}
  & \beta_{0,0,m_z}^1 = - m_z^2 \pi^2, \qquad
  && \phi_{0,0,m_z}^1 = [0,0,0,\sqrt{2} \sin m_z \pi z], \\
  & \beta^2_{0,0,m_z} = -\Pr m_z^2 \pi^2 - \frac{1}{\Ro}i, \qquad
  && \phi^2_{0,0,m_z} = \left[1, i , 0, \frac{-i \Lambda/\Fr^2}{i + m_z^2 \pi^2(\Pr-1)\Ro}\right]\sin m_z \pi z, \\
  & \beta^3_{0,0,m_z} = \overline{\beta_{0,0,m_z}^2}, \qquad
  && \phi^3_{0,0,m_z} = \overline{\phi_{0,0,m_z}^2}.
\end{aligned}
\end{equation}

{\it Case $\alpha \ne 0$:} In this case \eqref{pde eig prob} reduces to
\begin{align}
  & \Pr (D^2 - \alpha^2)C - \Lambda i \alpha_x z C -\frac{1}{\Ro} DW = \beta C, \label{ode eig prob 1} \\
  & \frac{1}{\Ro} DC + \Pr(D^2 - \alpha^2)^2 W - \Lambda \alpha^2 W - \Lambda i \alpha_x \left( D(z DW) - \alpha^2 z W \right) + \frac{\alpha^2}{\Fr^2} R  = \beta (D^2 - \alpha^2)W,  \label{ode eig prob 2}\\
  & - \frac{\Lambda}{\Ro \Fr^2 \alpha^2} (i \alpha_x C + i \alpha_y DW) + (D^2 - \alpha^2)R = \beta R,  \label{ode eig prob 3}\\
  & C = W = D W = R = 0 \qquad \text{at } z=0,1 \label{ode eig prob 4}.
\end{align}
Here $C$ is the vertical component of the vorticity,
\begin{equation} \label{UVCW}
  \begin{aligned}
    & - i \alpha_x V+ i \alpha_y U = C, \\
    & i \alpha_x U + i \alpha_y V = -DW.
  \end{aligned}
\end{equation}
The second condition in \eqref{UVCW} is just a consequence of incompressibility. This condition also imposes the boundary condition $DW=0$ at $z=0,1$ in \eqref{ode eig prob 4}. Note that $U$ and $V$ can be recovered from $C$ and $DW$ by the equations \eqref{UVCW}.




\subsection{Adjoint linear equations}
To describe the transition number we will also need the eigenvectors of the adjoint linear eigenvalue problem, given by
\begin{equation} \label{adjoint pde eig prob}
  \begin{aligned}
    & \Pr \Delta \bfu^{\ast} + \Lambda z \frac{\partial \bfu^{\ast}}{\partial x} - \Lambda w^{\ast} \hat{e}_3  + \frac{1}{\Ro}\hat{e}_3 \times \bfu^{\ast} - \nabla p^{\ast} - \frac{\Lambda}{\Ro \Fr^2}\rho^{\ast} \hat{e}_2 = \overline{\beta} \bfu^{\ast}, \\
    & -\frac{1}{\Fr^2} w^{\ast} + \Delta \rho^{\ast} = \overline{\beta} \rho^{\ast}, \\
    & \nabla \cdot \bfu^{\ast} = 0.
  \end{aligned}
\end{equation}
This eigenvalue problem can be solved with the same ansatz \eqref{ansatz}. The equations in the $\alpha=0$ case are
\begin{align*}
  & \Pr D^2 U - \frac{1}{\Ro} V = \beta U, \\
  & \Pr D^2 V + \frac{1}{\Ro} U - \frac{\Lambda}{\Ro \Fr^2} R = \beta V, \\
  & D^2 R = \beta R.
\end{align*}
The corresponding eigenvalues for $\beta_{0,0,m_z}^{k \ast} = \overline{\beta_{0,0,m_z}^{k}}$, $k=1,2,3$ and eigenvectors are:
\begin{equation} \label{adjoint alpha=0 eigenpair1}
  \begin{aligned}
  & \phi^{1 \ast}_{0,0,m_z} = [\Lambda, m_z^2 \pi^2(1-\Pr)\Ro \Lambda, 0, \Fr^2(1+m^4\pi^4(1-\Pr)^2\Ro^2]\sin m_z \pi z, \\
  & \phi^{2 \ast}_{0,0,m_z} = [1, -i , 0, 0]\sin m_z \pi z, \\
  & \phi^{3 \ast}_{0,0,m_z} = \overline{\phi_{0,0,m_z}^2}.
\end{aligned}
\end{equation}

For $\alpha \ne 0$, the equations reduce to
\begin{equation}
\label{adj ode eig prob}
 \begin{aligned}
&    \Pr (D^2 - \alpha^2)C^{\ast} + \Lambda i \alpha_x z C^{\ast} + \frac{1}{\Ro} DW^{\ast} + \frac{\Lambda}{\Ro \Fr^2} i \alpha_x R^{\ast}  = \overline{\beta} C^{\ast}, \\
  &  -\frac{1}{\Ro} DC^{\ast} + \Pr(D^2 - \alpha^2)^2 W^{\ast} - \Lambda \alpha^2 W^{\ast} +   \Lambda i \alpha_x \left( D(z DW^{\ast}) - \alpha^2 z W^{\ast} \right) \\
  & \qquad  \qquad  + \frac{\Lambda i \alpha_y}{\Ro \Fr^2} DR^{\ast}   = \overline{\beta} (D^2 - \alpha^2)W^{\ast},  \\
&     -\frac{1}{\Fr^2} W^{\ast} + (D^2 - \alpha^2)R^{\ast}  = \overline{\beta} R^{\ast}, \\
&    C^{\ast} = W^{\ast} = D W^{\ast} = R^{\ast} = 0 \qquad \text{at } z=0,1.
  \end{aligned}
\end{equation}

\subsection{Principle of exchange of stabilities (PES)}
We will assume that the eigenvalues are ordered so that
\[
   \Re \beta_{m_x, m_y, 1}(\lambda) \ge \Re \beta_{m_x, m_y, 2}(\lambda) \ge \cdots,
   \qquad \text{ for each } (m_x, m_y) \in \mathbb{Z}\times\mathbb{Z},
\]
where $\lambda = (\Fr, \Lambda, \Ro, \Pr)$.

By \eqref{alpha=0 eigenpairs}, the eigenvalues have always negative real part when $m_x = m_y = 0$.
Due to the horizontal periodicity of the problem, generically, an eigenvalue with nonzero horizontal wave number is either real with multiplicity two or is simple, non-real and its complex conjugate is also an eigenvalue.

Let
\begin{equation} \label{critical wave}
\alpha_c=\sqrt{(\alpha^c_x)^2 + (\alpha_y^c)^2}, \qquad
(\alpha_x^c, \alpha_y^c)=\left(\frac{m_x^c \pi}{L_x}, \frac{m_y^c \pi}{L_y}\right),
\end{equation}
be the critical horizontal wave number so that $\Re \beta_{m_x^c, m_y^c, 1}(\lambda)$ is the largest among all $\Re \beta_{m_x, m_y, m_z}(\lambda)$.

For our transition theorem, we will assume the existence of a critical parameter
\[
  \Lambda_c = \Lambda_c(L_x, L_y, \Fr, \Pr, \Ro),
\]
and integers $(m_x^c, m_y^c) \in \mathbb{Z}\times \mathbb{Z} \setminus \{(0, 0) \}$ such that the following PES condition is satisfied:
\begin{equation} \label{PES}
  \begin{aligned}
    & \Re \beta_{m_x^c, m_y^c, 1} = \Re \beta_{m_x^c, m_y^c, 2}
    \begin{cases}
      < 0 &  \text{ if }  \Lambda < \Lambda_c, \\
      = 0 &\text{ if }  \Lambda = \Lambda_c, \\
      > 0 &\text{ if }  \Lambda > \Lambda_c,
    \end{cases} \\
    & \Re \beta_{m_x, m_y, m_z} (\Lambda_c) < 0 \qquad   \text{ for } (m_x, m_y) \ne (m_x^c, m_y^c) \text{ or } m_z \notin \{ 1, 2 \}.
  \end{aligned}
\end{equation}
Note that the condition \eqref{Nonlinear Stability condition} assures that the eigenvalues lie on the left complex half plane if $\Lambda$ is small enough. In the last section, we will show the existence of parameter regimes where the PES \eqref{PES} is satisfied numerically.

By a general dynamic transition theorem \cite[Theorem 2.1.5]{ptd}, as $\Lambda$ crosses $\Lambda_c$,  the system always undergoes a dynamic transition from the basic shear flow to one of the three types (continuous, catastrophic or random). The types and structure of the dynamic transition are then dictated by the nonlinear interactions of different modes, and the remaining part of the paper is devoted to detailed characterization of the dynamic transition.

\section{Nonlinear Stability}
The following theorem gives a sufficient condition on the local nonlinear stability of the basic solution, in terms of the system parameters and the wave number of the critical mode. The condition is not optimal and may further be enhanced.

We will denote the eigenvalue with the largest real part as the critical eigenvalue, the corresponding eigenvectors as the critical eigenvectors and corresponding wave number as the critical wavenumber.
\begin{theorem} \label{Thm: Nonlinear Stability}
  Let $(\alpha_x^c, \alpha_y^c) = (\frac{m_x^c \pi}{L_x}, \frac{m_y^c \pi}{L_y})$ be the critical wave number as given in \eqref{critical wave}. If
  \begin{equation} \label{Nonlinear Stability condition}
    \alpha_c^2 := (\alpha_x^c)^2 + (\alpha_y^c)^2 > \left( \frac{1}{\text{\rm Fr}^4} + \frac{\sqrt{2}}{\text{\rm Ro}} \right)
    \max \left\{\left(\frac{\Lambda^2}{\Pr}\right)^{1/3}, \frac{1}{\Pr} \right\},
  \end{equation}
  then the system is locally nonlinearly stable.
\end{theorem}
\begin{proof}
We will prove the linear stability of the basic solution by showing that the spectrum of the linear operator is in the left half complex plane bounded away from the imaginary axis.
Since the linear operator consists of the Laplacian plus lower order terms, it is sectorial.
Hence the linear stability of the basic solution implies its local nonlinear stability by Theorem 5.1.1 in \cite{henry}.

Thus we need to show that the real part of the critical eigenvalue is negative. For this, consider the critical eigenvalue and eigenvector of the problem \eqref{ode eig prob 1}--\eqref{ode eig prob 4}.

 Here we also use  $\norm{\cdot}$ for the  $L^2(0,1)$ norm. Then taking $L^2$ inner product  between  \eqref{ode eig prob 1}--\eqref{ode eig prob 3} and $(C, W, R)$, we obtain that
\begin{equation} \label{linear stability proof eq1}
    -\mathcal{E}_1 - \Lambda i \alpha^c_x \mathcal{E}_2 + \mathcal{I} = \beta \mathcal{E}_3,
  \end{equation}
  where
\begin{align*}
&
 \mathcal{E}_1 = \Pr (\norm{DC}^2 + \alpha_c^2 \norm{C}^2) + \Pr (\norm{D^2 W}^2 +2\alpha_c^2 \norm{DW}^2 + \alpha_c^4 \norm{W}^4) \\
 & \qquad \qquad + \Lambda \alpha_c^2 \norm{W}^2 + \norm{DR}^2 + \alpha_c^2 \norm{R}^2, \\
 &
    \mathcal{E}_2 = \norm{z^{1/2} C}^2 + \norm{z^{1/2}DW}^2 + \alpha_c^2 \norm{z^{1/2}W}^2, \\
 &
    \mathcal{E}_3 = \norm{C}^2 + \norm{DW}^2 + \alpha_c^2 \norm{W}^2 + \norm{R}^2, \\
  &
   \mathcal{I} = \int_0^1 \left( -\frac{1}{\Ro} DW \overline{C} - \frac{1}{\Ro} DC \overline{W} - \frac{\alpha_c^2}{\Fr^2}  R \overline{W} - i A \alpha^c_x C \overline{R} + i \mu \alpha^c_y DW \overline{R} \right) dz, \\
   &\mu = \dfrac{\Lambda}{\Ro \Fr^2 \alpha_c^2}.
\end{align*}

The real part of \eqref{linear stability proof eq1} implies that
  \[
    \Re \beta = \frac{-\mathcal{E}_1 + \Re \mathcal{I}}{\mathcal{E}_3}.
  \]
The proof will be done if we can show that $\abs{\mathcal{I}} < \mathcal{E}_1$ which implies $\Re \beta <0$.

  We estimate $\abs{\mathcal{I}}$ as
  \begin{equation} \label{linear stability proof eq2}
    \abs{\mathcal{I}} \le \Pr \norm{DC}^2 + \norm{DR}^2 + \alpha_c^2 \norm{R}^2 + \frac{\mu^2}{2} \norm{C}^2 + \left( \frac{1}{\Ro^2 \Pr}  + \frac{\alpha_c^2}{2 \Fr^4} +\frac{\mu^2 \alpha_c^2}{4} \right) \norm{W}^2.
  \end{equation}
In view of  \eqref{linear stability proof eq1} with \eqref{linear stability proof eq2}, the condition $\abs{\mathcal{I}} < \mathcal{E}_1$ will be satisfied if
  \begin{equation}\label{linear stability proof two conditions}
    \frac{1}{\Ro^2 \Pr}  + \frac{\alpha_c^2}{2 \Fr^4} +\frac{\mu^2 \alpha_c^2}{4} \le \Pr \alpha_c^4 \qquad \text{and} \qquad
    \frac{\mu^2}{2} \le \Pr \alpha_c^2.
  \end{equation}

Now we use the Young's inequality with $p=3/2$ and $q=3$ in the second inequality below
  \[
    \left(\frac{1}{\Fr^4} + \frac{\sqrt{2}}{\Ro}  \right)^{3} \ge
    \left(\frac{1}{3 \Fr^4} + \frac{2}{3 \Ro}  \right)^{3} \ge
    \frac{1}{\Ro^2} \frac{1}{\Fr^4}.
  \]
  By \eqref{Nonlinear Stability condition} and the above estimate, we deduce that
  \[
    \frac{\mu^2}{2} =
    \frac{\Lambda^2 \alpha_c^2}{2 \Ro^2 \Fr^4 \alpha_c^6} \le
    \frac{\Lambda^2 \alpha_c^2}{2 \Ro^2 \Fr^4} \left(\frac{1}{\Fr^4} + \frac{\sqrt{2}}{\Ro}  \right)^{-3} \frac{\Pr}{\Lambda^2} \le
    \Pr \alpha_c^2
  \]
  and, consequently, the second condition in \eqref{linear stability proof two conditions} is satisfied.

  Using the second condition in \eqref{linear stability proof two conditions}, the first condition in \eqref{linear stability proof two conditions} will be satisfied if
  \begin{equation}\label{linear stability proof one cond}
    \Pr \alpha_c^4 - \frac{\alpha_c^2}{\Fr^4} - \frac{2}{\Ro^2 \Pr}  \ge 0.
  \end{equation}
    By \eqref{Nonlinear Stability condition}, and using $(a+b)^2 \ge a^2 + b^2$ for positive $a$, $b$, we have
  \[
    \Pr \alpha_c^2 \ge
    \frac{1}{\Fr^4} + \frac{\sqrt{2}}{\Ro} \ge
    \frac{1}{2\Fr^4} + \frac{1}{2} \sqrt{\frac{1}{\Fr^8} + \frac{8}{\Ro^2}},
  \]
  which implies that \eqref{linear stability proof one cond} indeed holds. The proof is complete.
\end{proof}

\begin{remark}
  A sufficient condition for the local nonlinear stability of the system can also be stated in terms of the minimum wave number given by
  \[
    \alpha_{\min} = \min \left \{\sqrt{\left( \frac{m_x \pi}{L_x} \right)^{2} + \left( \frac{m_y \pi}{L_y} \right)^{2} }: (m_x, m_y)\in  \ \mathbb{Z} \times \mathbb{Z}, m_x^2 + m_y^2 \ne 0 \right\}
    = \min \left\{ \frac{\pi}{L_x}, \frac{\pi}{L_y}\right\}.
  \]
  Noticing that $\alpha_{\min} \le \alpha_c$, it follows that if $\alpha_{\min}$ satisfies \eqref{Nonlinear Stability condition} then the system is locally nonlinearly stable.
\end{remark}

\section{Main Dynamic Transition Theorem}

As mentioned earlier, by a general dynamic transition theorem \cite[Theorem 2.1.5]{ptd}, as $\Lambda$ crosses $\Lambda_c$,  the system always undergoes a dynamic transition from the basic shear flow to one of the three types (continuous, catastrophic or random). In this section, we shall demonstrate that the types and the structure of the dynamic transitions are dictated by the sign of the transition number to be defined below, which captures the nonlinear interactions of different modes. To this end, from now on we use the following notations for the eigenvalues and eigenvectors.
\[
  \begin{aligned}
    & \phi_1 = \phi_{m_x^c, m_y^c, 1}, \quad
    \phi_{0, m_z} = \phi_{0, 0, m_z}, \quad
    \phi_{2, m_z} = \phi_{2 m_x^c, 2 m_y^c, m_z}, \\
    & \phi_1^{\ast} = \phi^{\ast}_{m_x^c, m_y^c, 1}, \quad
    \phi^{\ast}_{0, m_z} = \phi^{\ast}_{0, 0, m_z}, \quad
    \phi_{2, m_z}^{\ast} = \phi^{\ast}_{2 m_x^c, 2 m_y^c, m_z}, \\
    & \beta_1 = \beta_{m_x^c, m_y^c, 1}, \quad
    \beta_{0, m_z} = \beta_{0, 0, m_z}, \quad
    \beta_{2, m_z} = \beta_{2 m_x^c, 2 m_y^c, m_z}.
  \end{aligned}
\]

We then define the transition number by
\begin{equation} \label{A}
  A = \sum_{m_z=1}^{\infty} A_{0,m_z} + A_{2, m_z},
\end{equation}
where
\begin{equation} \label{A0 A2}
  \begin{aligned}
    & A_{0, m_z} = \frac{1}{\langle \phi_1, \phi_1^{\ast} \rangle} \Phi_{0,m_z} \langle G_s(\phi_1, \phi_{0,m_z}), \phi_1^{\ast} \rangle,  \\
    & A_{2, m_z} = \frac{1}{\langle \phi_1, \phi_1^{\ast} \rangle} \Phi_{2, m_z} \langle G_s(\overline{\phi_1}, \phi_{2, m_z}), \phi_1^{\ast} \rangle, \\
    & \Phi_{0,m_z} = \frac{1}{-\beta_{0,m_z} \langle \phi_{0,m_z}, \phi_{0,m_z}^{\ast} \rangle} \langle G_s(\phi_1, \overline{\phi_1}), \phi_{0,m_z}^{\ast} \rangle, \\
    & \Phi_{2, m_z} = \frac{1}{(2 \beta_1 -\beta_{2, m_z}) \langle \phi_{2, m_z}, \phi_{2, m_z}^{\ast} \rangle} \langle G(\phi_1, \phi_1), \phi_{2, m_z}^{\ast} \rangle. \\
  \end{aligned}
\end{equation}
Here  $G_s(\varphi_1, \varphi_2) = G(\varphi_1, \varphi_2) + G(\varphi_2, \varphi_1) $ and $G((\bfu_1, \rho_1), (\bfu_2, \rho_2) ) = -\mathcal{P} \left( \bfu_1 \cdot \nabla (\bfu_2, \rho_2) \right)$; see \eqref{operators}.

The numbers $A_{0, j}$ and $A_{2, j}$ represent the nonlinear interactions of the critical modes with the modes having wave numbers $0$ and $2 \alpha_{c}$ respectively.

We note that the transition number $A$ defined by \eqref{A} is valid for the critical crossing of both the complex and real eigenvalues, with the only difference being that $\Im(\beta_1) = 0$ in the last expression in \eqref{A0 A2} in the critical crossing of real eigenvalues.

\medskip

The main dynamic transition theorem of the paper is as follows:

\begin{theorem}\label{main thm}
Assume that the PES condition \eqref{PES} is valid.
  \begin{enumerate}
    \item If $\Im(\beta_1) = 0$ near $\Lambda = \Lambda_c$, then the transition number $A$ defined by \eqref{A}  is real and the following statements hold true:
    \begin{enumerate}
      \item If $A<0$, then  the original system \eqref{main equs} with \eqref{BC} undergoes a continuous dynamic transition on $\Lambda > \Lambda_c$ and bifurcates, on $\Lambda > \Lambda_c$, from the basic shear flow to an attracting circle of steady states:
      \[
        \Sigma_{\Lambda} = \left \{2 \sqrt{\frac{-\beta_1}{A}} \Re (e^{i \gamma} \phi_1) + {\rm O}(\abs{\beta_1}) : \gamma \in \mathbb{R} \right \}.
      \]
      \item If $A>0$,  then the system \eqref{main equs} with \eqref{BC} undergoes a jump transition,  and bifurcates, on $\Lambda < \Lambda_c$, from the basic shear flow to a repeller $\Sigma_{\Lambda}$ having the same form as given above.
    \end{enumerate}
    \item If $\Im(\beta_1) \not= 0$ near $\Lambda = \Lambda_c$, then the transition number $A$ is non-real and the following statements hold true:
    \begin{enumerate}
      \item If $\Re(A) < 0$, then the system undergoes a  continuous transition, and bifurcates, on  $\Lambda > \Lambda_c$, from the basic shear flow  to a stable limit cycle given by
      \begin{equation} \label{bifurcated solution}
        \phi_{\text{bif}} = z(t) \phi_1 + \overline{z(t) \phi_1} + O(-\Re(\beta_1)),
      \end{equation}
      where
      \begin{equation} \label{ode bifurcated solution}
        z(t) = \sqrt{\frac{-\Re(\beta_1)}{\Re(A)}} \exp \left(i \Im(\beta_1) t \right).
      \end{equation}
      \item If $\Re(A) > 0$, then the system undergoes a jump transition, and bifurcates, on  $\Lambda <\Lambda_c$, from the basic shear flow to an unstable periodic solution having the same form as \eqref{bifurcated solution} and \eqref{ode bifurcated solution}.
    \end{enumerate}
  \end{enumerate}
\end{theorem}

\section{Proof of \autoref{main thm}.}
By the dynamic transition theory \cite{ptd}, the dynamical transition is fully captured by the reduced system of the original system to the center manifold associated with the unstable modes. Hence the central gravity of the proof is to derive the reduced system of \eqref{main equs}--\eqref{BC}. For this purpose, we derive in Appendix A an approximation formula for the center manifold function adapted to the specific structure of the underlying system, and use it to prove the main theorem in this section.

Let $E_1$ be the finite dimensional vector space spanned by the first critical eigenvectors, and $E_2$ be its orthogonal complement as as defined in \eqref{E1E2}. Then any $\Phi \in E_1$ can be written as
\begin{equation} \label{center part 2}
  \Psi = z(t) \phi_1 + \overline{z(t) \phi_1} \in E_1.
\end{equation}
Now we will utilize \eqref{new cm formula} to approximate the center manifold. First notice that
\begin{equation}\label{pr1}
  \ip{G(\phi_1, \phi_1), e^{-i \alpha \cdot x} \varphi(z)} = 0, \qquad \text{if } \alpha \ne 2 \alpha_c = (2 m_x^c, 2m_y^c)
\end{equation}
due to the orthogonality $\int_0^{L_x}\int_0^{L_y} e^{i \alpha_c \cdot x} e^{i \alpha_c \cdot x} e^{-i \alpha \cdot x} dy dx = 0$ if $\alpha \ne 2 \alpha_c$. Similarly
\begin{equation}\label{pr2}
  \ip{G(\phi_1, \overline{\phi_1}), e^{-i \alpha \cdot x} f(z)} = 0, \qquad \text{if } \alpha \ne (0, 0).
\end{equation}
due to the orthogonality $\int_0^{L_x}\int_0^{L_y} e^{i \alpha_c \cdot x} e^{-i \alpha_c \cdot x} e^{-i \alpha \cdot x} dy dx = 0$ if $\alpha \ne (0, 0)$.
By \eqref{pr1} and \eqref{pr2}, we find that
\begin{equation}\label{pr3}
  P_2 G(\phi_1, \phi_1) = G(\phi_1, \phi_1), \qquad
  P_2 G(\phi_1, \overline{\phi_1}) = G(\phi_1, \overline{\phi_1}),
\end{equation}
where $P_2$ is the projection onto $E_2$. Moreover, \eqref{pr1} and \eqref{pr2} also imply that
\begin{equation}\label{pr4}
  \ip{G(\phi_1, \phi_1), \phi_{m_x, m_y, m_z}} = 0, \qquad \text{if } (m_x, m_y) \ne (2m_x^c, 2m_y^c).
\end{equation}
and
\begin{equation}\label{pr5}
  \ip{G(\phi_1, \overline{\phi_1}), \phi_{m_x, m_y, m_z}} = 0, \qquad \text{if } (m_x, m_y) \ne (0, 0).
\end{equation}
Now using \eqref{pr3}, \eqref{pr4} and \eqref{pr4} in \eqref{new cm formula}, we find that the center manifold has the quadratic approximation
\begin{equation} \label{CM expansion complex}
  \Phi = \sum_{j=1}^{\infty} z^2 \Phi_{2, j} \phi_{2, j} + \abs{z}^2 \Phi_{0, j} \phi_{0, j}+ \overline{z^2 \Phi_{2, j} \phi_{2, j}} + o(2),
\end{equation}
where the coefficients $\Phi_{2, j}$ and $\Phi_{0, j}$ are as defined in \eqref{A0 A2}.
Letting $\phi = \Psi + \Phi$ which amounts to considering the dynamics on the center manifold, then taking projection onto the center-unstable space and finally using $\ip{\overline{\phi_1},\phi_1^{\ast}}=0$,  we get
\begin{equation} \label{pre-reduced}
  \frac{dz}{dt} = \beta_1 z + \frac{1}{\langle \phi_1, \phi_1^{\ast} \rangle} \ip{G(\Psi + \Phi), \phi_1^{\ast}}.
\end{equation}
The nonlinear term in \eqref{pre-reduced} can be written as
\[
  \ip{G(\Psi + \Phi), \phi_1^{\ast}} = \sum_{j=1}^{\infty}
  \ip{G_s(z \phi_1, \Phi_{0, j} \phi_{0, j} \abs{z}^2 ), \phi_1^{\ast}} +
  \ip{G_s(\overline{z \phi_1}, \Phi_{2, j} \phi_{2, j} z^2), \phi_1^{\ast}} + o(3).
\]
Then the equation \eqref{pre-reduced} becomes
\begin{equation} \label{reduced equation real}
  \frac{dz}{dt} = \beta_1 z + A \abs{z}^2 z + o(3),
\end{equation}
where $A$ is defined by \eqref{A}.
In polar coordinates $z(t) = \abs{z}e^{i \gamma}$, \eqref{reduced equation real} is equivalent to
\begin{equation}
\begin{aligned} \label{reduced equation complex polar}
&  \frac{d\abs{z}}{dt} = \Re(\beta_1) \abs{z} + \Re(A) \abs{z}^3 + o(\abs{z}^3), \\
&  \frac{d\gamma}{dt} = \Im(\beta_1) + \Im(A) \abs{z}^2 + o(\abs{z}^3).
\end{aligned}
\end{equation}
To finalize the proof, thanks to the dynamical transition theorems \cite{ptd}, there remains to analyze the stability of the zero solution of \eqref{reduced equation complex polar}.

In the case where $\Im \beta_1 = 0$, if one carries out this reduction procedure using real eigenvectors $e_1 = \Re \phi_1$, $e_2 = \Im \phi_1$ then as all the terms determining the transition number $A$ will be real, it must follow that the coefficient $A$ must also be real. Supposing $\beta_1>0$ and $A<0$, it follows from the first equation in \eqref{reduced equation complex polar} that $\abs{z(t)} \to -\sqrt{\frac{-\beta_1}{A}}$ as $t \to \infty$ and the solutions of \eqref{reduced equation real} will evolve to a circle of steady state solutions.

In the case $\Im \beta_1 \ne 0$, $A$ is non-real, and the equation \eqref{reduced equation real} is the normal form of the classical Hopf bifurcation.
For $\Re(A)<0$, a periodic solution of \eqref{reduced equation complex polar} will bifurcate on $\Re \beta_1>0$ given by \eqref{ode bifurcated solution}.
If $\Re(A)>0$ then the system will evolve to an attractor far away from $z=0$ on $\Re \beta_1>0$.

\section{Numerical Strategy}
\subsection{Numerical Solution of the Eigenvalue Problem}
To solve \eqref{ode eig prob 1}--\eqref{ode eig prob 4}, we use the spectral Legendre Galerkin method. For this, we first make the change of variable $\tilde{z} = 2z - 1$, which transforms the z-domain from $[0,1]$ to $[-1,1]$ and the derivative  $D = \frac{d}{dz} = 2 \frac{d}{d \tilde{z}} = 2 \tilde{D}$. Next, we discretize the problem using the expansions
\begin{equation} \label{expansions}
  C = \sum_{i=0}^{N-1} \widehat{C}_i e^C_i(\tilde{z}), \qquad
  W = \sum_{i=0}^{N-1} \widehat{W}_i e^W_i(\tilde{z}), \qquad
  R = \sum_{i=0}^{N-1} \widehat{R}_i e^R_i(\tilde{z}).
\end{equation}

The basis functions are the generalized Jacobi polynomials (compact combinations of Legendre polynomials) given by
\[
  \begin{aligned}
    & e^C_i(\tilde{z}) = e^R_i(\tilde{z}) = L_i(\tilde{z}) - \frac{i(1+i)}{(2+i)(3+i)} L_{i+2}(\tilde{z}), \\
    & e^W_i(\tilde{z}) = L_i(\tilde{z}) + a_i L_{i+2}(\tilde{z}) + b_i L_{i+4}(\tilde{z}), \\
    & a_i = -\frac{2(5+2i)(9+5i+i^2)}{(3+i)(4+i)(7+2i)}, \qquad
    b_i = -1 - a_i.
  \end{aligned}
\]
The basis functions satisfy the boundary conditions
\[
  e^C_i(\tilde{z}) = e^W_i(\tilde{z}) = D^2e^W_i(\tilde{z}) = e^R_i(\tilde{z}) = 0, \quad \tilde{z}=-1,1.
\]

The discretization of the linear eigenvalue problem is done by plugging the expansions \eqref{expansions} in \eqref{ode eig prob 1}--\eqref{ode eig prob 4} and taking the inner products with basis functions. Then the resulting finite dimensional matrix eigenvalue problem is solved.

\subsection{The Numerical Computation of the Transition Number}
We now give a more detailed form of the transition number $A$ given by \eqref{A} and described by the equations \eqref{A0 A2}. We represent the eigenvectors
\begin{align*}
& \phi_1 = e^{i \alpha_c \cdot x} \psi_1(z), \qquad && \psi_1(z) = [u_1, v_1, w_1, r_1]^T, \\
& \phi_{2,m_z} = e^{2 i \alpha_c \cdot x} \psi_{2,m_z}(z), && \psi_{2,m_z}(z) = [u_{2,m_z}, v_{2,m_z}, w_{2,m_z}, r_{2,m_z}]^T, \\
& \phi_{0,m_z} = e^{2 i \alpha_c \cdot x} \psi_{0,m_z}(z), && \psi_{0,m_z}(z) = [u_{0,m_z}, v_{0,m_z}, w_{0,m_z}, r_{0,m_z}]^T,
\end{align*}
and
\[
  \phi_1^* = e^{i \alpha_c \cdot x} \psi_1^*(z), \qquad
  \phi_{0, m_z}^* = e^{i \alpha_c \cdot x} \psi_{0, m_z}^*(z), \qquad
  \phi_{2, m_z}^* = e^{i \alpha_c \cdot x} \psi_{2, m_z}^*(z).
\]
By using
\[
  i \alpha_x u_1 + i \alpha_y v_1 + D w_1 = 0, \qquad 2i \alpha_x u_{2, m_z} + 2i \alpha_y v_{2, m_z} + Dw_{2, m_z} = 0,
\]
the expressions in \eqref{A0 A2} can be written as
\begin{equation} \label{pre expanded A0A2}
  \begin{aligned}
  & G(\phi_1, \phi_{0,m_z}) = e^{i \alpha_c \cdot x} (u_1 \partial_x + v_1 \partial_y + w_1 D) \phi_{0,m_z} = e^{i \alpha_c \cdot x} w_1 D \psi_{0,m_z}, \\
  & G(\phi_{0,m_z}, \phi_1) = (u_{0,m_z} \partial_x + v_{0,m_z} \partial_y) \phi_1 = e^{i \alpha_c \cdot x} (i \alpha_x u_{0,m_z} + i \alpha_y v_{0,m_z} \partial_y) \psi_1, \\
  & G(\overline{\phi_1}, \phi_{2, m_z}) = e^{i \alpha_c \cdot x} (2 D \overline{w_1} + \overline{w_1} D) \psi_{2, m_z}, \quad
  G(\phi_{2, m_z}, \overline{\phi_1}) = e^{i \alpha_c \cdot x} (\frac{1}{2} D w_{2, m_z} + w_{2, m_z} D) \overline{\psi_{1}}, \\
  & G(\phi_1, \overline{\phi_1}) = (Dw_1 + w_1 D) \overline{\psi_1}, \quad
  G(\phi_1, \phi_1) = e^{2 i \alpha_c \cdot x} (-Dw_1 + w_1 D) \psi_1.
  \end{aligned}
\end{equation}
Hence the transition number becomes
\begin{equation} \label{A0 A2 expanded}
  \begin{aligned}
    & A_{0, m_z} = - \frac{\Phi_{0,m_z}}{\int_0^1 dz \psi_1 \cdot \overline{\psi_1^{\ast}}} \int_0^1 dz \left( w_1 D\psi_{0,m_z} + \left( i \alpha_x u_{0,m_z} + i \alpha_y v_{0,m_z} \right) \psi_1 \right) \cdot \overline{\psi_{1}^{\ast}}, \\
    & A_{2, m_z} = - \frac{\Phi_{2, m_z}}{\int_0^1 dz \psi_1 \cdot \overline{\psi_1^{\ast}}}  \int_0^1 \left( 2D \overline{w_1} \psi_{2, m_z} + \overline{w_1} D \psi_{2, m_z} + \frac{1}{2} Dw_{2, m_z} \overline{\psi_1} + w_{2, m_z} D \overline{\psi_1} \right) \cdot \overline{\psi_1^{\ast}}, \\
    & \Phi_{0,m_z} = \frac{1}{ \beta_{0,m_z} \int_0^1 dz \psi_{0,m_z} \cdot \overline{\psi_{0,m_z}^{\ast}} }  \int_0^1 dz 2 \Re \left\{ \left( Dw_1 + w_1 D \right)\overline{\psi_1} \right\} \cdot \overline{\psi_{0,m_z}^{\ast k}}, \\
    & \Phi_{2, m_z} = \frac{-1}{(2 \beta_1 -\beta_{2, m_z}) \int_0^1 dz \psi_{2,m_z} \cdot \overline{\psi_{2,m_z}^{\ast}}}\int_0^1 \left( -Dw_1 + w_1 D \right)\psi_1 \cdot \overline{\psi_{2,m_z}^{\ast}}.
  \end{aligned}
\end{equation}

The numerical strategy is as follows. Initially we fix all parameters and compute the eigenvalues for the first few wave indices $(m_x^c, m_y^c)$ numerically, and call the minimum of these eigenvalue $\beta_1$. Next we vary one of the parameters, say $\Lambda$, to determine the critical value of $\Lambda_c$, so that $\Re \beta_1(\Lambda_c) = 0$. Then we fix $\Lambda =\Lambda_c$ and find all the eigenpairs corresponding to wave indices $(0,0)$ and $2(m_x^c, m_y^c)$ and compute the transition number using \eqref{A0 A2 expanded}.

\section{Numerical Investigations}
\subsection{Parameters vs Critical Eigenvalue}
According to \eqref{Nonlinear Stability condition}, increasing the parameters $\Fr$, $\Ro$, $\Pr$ has a stabilizing effect while increasing $\Lambda$ has a destabilizing effect on the eigenvalue $\beta_1$ with largest real part. We verify this numerically in \autoref{fig:Fr Ro Pr Lambda vs beta1} where the first real part of the eigenvalue is plotted against the system parameters for fixed wave number $\alpha$. The first and last plots of \autoref{fig:Fr Ro Pr Lambda vs beta1} also show that the eigenvalue of largest real part is always bounded below by $\Re \beta_{0, 0, 1}^k$ ($k=1,2,3$).
\begin{figure}[H]
  \centering
  \begin{tikzpicture}
  \begin{axis}[
      xlabel=$\Fr$,
      ylabel=$\Re (\beta_1)$,
      height=3cm,
      width=4cm]
  \addplot[mark=*, mark size=1] coordinates { (0.08, 20.945) (0.083, 16.847) (0.086, 13.212) (0.089, 9.977) (0.092, 7.09) (0.095, 4.507) (0.098, 2.19) (0.101, 0.108) (0.104, -1.768) (0.107, -3.46) (0.11, -4.988) (0.113, -6.371) (0.116, -7.007) (0.119, -7.007) (0.122, -7.007) (0.125, -7.007) (0.128, -7.007) };
  \end{axis}
  \end{tikzpicture}
  \begin{tikzpicture}
  \begin{axis}[
      xlabel=$\Ro$,
      height=3cm,
      width=4cm]
  \addplot[mark=*, mark size=1] coordinates { (0.1, 135.46) (0.15, 102.447) (0.2, 82.931) (0.25, 69.712) (0.3, 60.023) (0.35, 52.541) (0.4, 46.547) (0.45, 41.61) (0.5, 37.457) (0.55, 33.903) (0.6, 30.819) (0.65, 28.112) (0.7, 25.711) (0.75, 23.565) (0.8, 21.632) (0.85, 19.88) (0.9, 18.282) (0.95, 16.819) (1., 15.473) };
  \end{axis}
  \end{tikzpicture}
  \begin{tikzpicture}
  \begin{axis}[
      xlabel=$\Pr$,
      height=3cm,
      width=4cm]
  \addplot[mark=*, mark size=1] coordinates { (0.1, 20.465) (0.15, 18.354) (0.2, 16.37) (0.25, 14.489) (0.3, 12.699) (0.35, 10.991) (0.4, 9.36) (0.45, 7.802) (0.5, 6.313) (0.55, 4.893) (0.6, 3.537) (0.65, 2.246) (0.7, 1.016) (0.75, -0.153) (0.8, -1.262) (0.85, -2.314) (0.9, -3.309) (0.95, -4.249) (1., -5.137) };
  \end{axis}
  \end{tikzpicture}
  \begin{tikzpicture}
  \begin{axis}[
      xlabel=$\Lambda$,
      height=3cm,
      width=4cm]
  \addplot[mark=*, mark size=1] coordinates { (0.3, -7.007) (0.35, -7.007) (0.4, -7.007) (0.45, -7.007) (0.5, -7.007) (0.55, -7.007) (0.6, -6.703) (0.65, -5.692) (0.7, -4.706) (0.75, -3.742) (0.8, -2.8) (0.85, -1.879) (0.9, -0.976) (0.95, -0.091) (1., 0.778) (1.05, 1.63) (1.1, 2.467) (1.15, 3.291) (1.2, 4.1) };
  \end{axis}
  \end{tikzpicture}
  \caption{The real part of the first eigenvalue when the system parameters are fixed $L_x=2$, $L_y=2$, $\Fr=0.1$, $\Lambda=1$, $\Pr=0.71$, $\Ro=2$ except for one of them which is varied.}
  \label{fig:Fr Ro Pr Lambda vs beta1}
\end{figure}
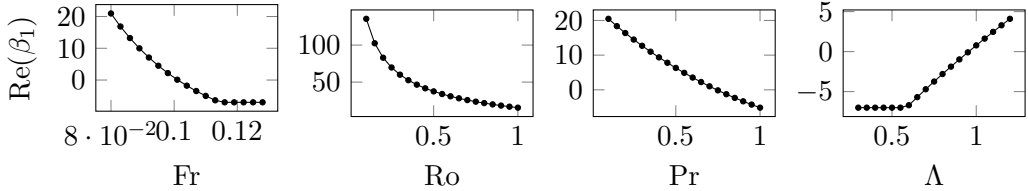

\subsection{Stability Diagrams}
\begin{figure}[H]
\centering
\begin{tikzpicture}
\begin{axis}[
  xlabel=$\Ro$,
  ylabel = $\Lambda$,
  height = 4cm,
  width = 8cm,
  ymode= log,
  xmode=log]
\addplot[blue, mark=*, mark size=1.5] table [] {
Ro  Lambda
0.05  0.451
0.1 0.799
0.15  1.166
0.2 1.536
0.25  1.906
0.3 2.274
0.35  2.64
0.4 3.004
0.45  3.365
0.5 3.724
}; \label{RoVsCritLambda1}
\addplot [red, mark=triangle*, mark size=1.5] table [] {
Ro  Lambda
0.01  0.089
0.015 0.113
0.02  0.145
0.025 0.185
0.03  0.233
0.035 0.29
0.04  0.355
}; \label{RoVsCritLambda2}
\addplot[black, mark=square*, mark size=1.5] table [] {
Ro  Lambda
0.01  0.001
0.028 0.003
0.046 0.005
0.064 0.007
0.082 0.009
0.1 0.011
1.3 0.144
1.1 0.122
0.9 0.1
0.7 0.077
0.5 0.055
}; \label{RoVsEnergyStabLambda}
\node[rotate=15] at (.04, 1.2) {\tiny Unstable};
\node[rotate=15] at (.08, .1) {\tiny Locally Nonlinearly Stable};
\node[rotate=15] at (.4, .005 ) {\tiny Energy Stable};
\end{axis}
\end{tikzpicture}
\caption{Stability curves for the basic shear flow for $\Pr = 0.71$, $\Fr = 0.2$, $L_x = L_y = 1$. The curve indicates transition to multiple equilibria (\ref{RoVsCritLambda1}), transition to spatio-temporal oscillations (\ref{RoVsCritLambda2}) and energy stability (\ref{RoVsEnergyStabLambda}).}
\label{fig:Ro vs Lambda}
\end{figure}
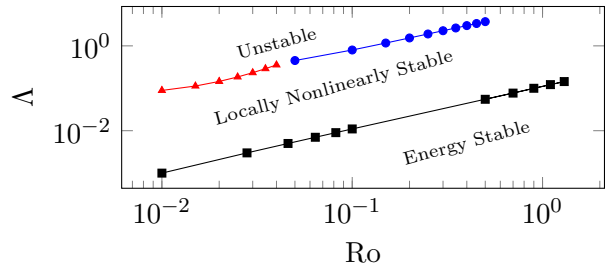

\begin{figure}[H]
  \centering
  \begin{tikzpicture}
  \begin{axis}[
      xlabel=$\Fr$,
      ylabel = $\Lambda$,
      height = 4cm,
      width = 8cm]
  \addplot[blue, mark=*, mark size=1.5] table [] {
    Fr  Lambda
    0.1 0.955
    0.11  1.393
    0.12  1.961
    0.13  2.682
    0.14  3.575
    0.15  4.66
  }; \label{FrVsCritLambda};
  \addplot[black, mark=square*, mark size=1.5] table [] {
    Fr  Lambda
    0.1 0.014
    0.13  0.04
    0.16  0.091
    0.19  0.18
    0.22  0.324
    0.25  0.54
    0.28  0.85
    0.31  1.277
    0.34  1.848
    0.37  2.592
    0.4 3.541
  }; \label{FrVsEnergyStabLambda};
  \node[rotate=20] at (0.1, 3) {\tiny Unstable};
  \node[rotate=20] at (0.25, 2.5) {\tiny Locally Nonlinearly Stable};
  \node[rotate=20] at (0.35, 1) {\tiny Energy Stable};
  \end{axis}
  \end{tikzpicture}
  \caption{Stability curves for the basic shear flow for $\Pr = 0.71$, $\Ro = 2$, $L_x = 1$ and $L_y = 1$. The curve indicates transition to multiple equilibria (\ref{FrVsCritLambda}) and energy stability (\ref{FrVsEnergyStabLambda}).}
  \label{fig:Fr vs Lambda}
\end{figure}

\begin{figure}[H]
  \centering
  \begin{tikzpicture}
  \begin{axis}[
      xlabel=$\Pr$,
      ylabel = $\Lambda$,
      ymax=10,
      height = 4cm,
      width = 8cm,
      ymode= log,
      xmode=log]
  \addplot[blue, mark=*, mark size=1.5] table [] {
    Pr  Lambda
    0.1 1.894
    0.11  2.081
    0.12  2.269
    0.13  2.456
    0.14  2.644
    0.15  2.832
    0.16  3.02
    0.17  3.208
    0.18  3.396
    0.19  3.584
    0.2 3.772
  }; \label{PrVsCritLambda}
  \addplot[black, mark=square*, mark size=1.5] table [] {
    Pr  Lambda
    0.1 0.031
    0.14  0.044
    0.18  0.056
    0.22  0.069
    0.26  0.081
    0.3 0.094
    0.34  0.106
    0.38  0.118
    0.42  0.131
    0.46  0.143
    0.5 0.156
  }; \label{PrVsEnergyStabLambda};
  \node[rotate=5] at (0.12, 5) {\tiny Unstable};
  \node[rotate=5] at (.2, .6) {\tiny Locally Nonlinearly Stable};
  \node[rotate=5] at (.3, .05) {\tiny Energy Stable};
  \end{axis}
  \end{tikzpicture}
  \caption{Stability curves for the basic shear flow for $\Fr = 0.2$, $\Ro = 2$, $L_x = L_y = 1$. The curve indicates transition to multiple equilibria (\ref{PrVsCritLambda}) and energy stability (\ref{PrVsEnergyStabLambda}).}
  \label{fig:Pr vs Lambda}
\end{figure}
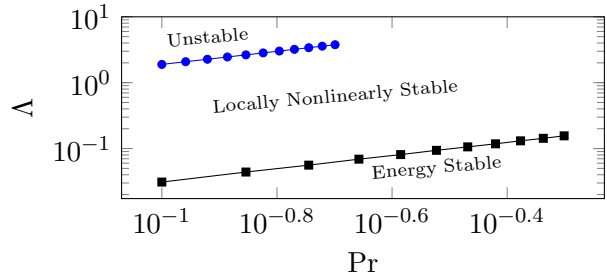

\autoref{fig:Ro vs Lambda} and \autoref{fig:Fr vs Lambda} show several stability regions in the parameter space. In both figures, a cubic basin $L_x=1$, $L_y=1$ is chosen and the Prandtl number is fixed at $\Pr=0.71$.

\autoref{fig:Ro vs Lambda} investigates the stability of the system in the $\Ro$-$\Lambda$ plane. The figure shows that the system has energy stability for small $\Lambda$. As $\Lambda$ increases, the system becomes nonlinearly stable and finally a dynamic transition occurs, either to spatiotemporal oscillations or to multiple equilibria depending on $\Ro$, as indicated in \autoref{main thm}. The curves leading to different transitions seem to intersect at a single point, where the reduced system will become four dimensional.

\autoref{fig:Fr vs Lambda} investigates the stability of the system in the $\Fr$-$\Lambda$ plane. The figure shows that once again the system has energy stability for small $\Lambda$. As $\Lambda$ increases, the system becomes nonlinearly stable and finally a dynamic transition occurs. In this case we find that the system exhibits only transitions to multiple equilibria depending on $\Ro$.

\subsection{The Transition Number}
We demonstrate some numerical details, namely the coefficients of the center manifold function and the nonlinear interactions contributing to the transition number in \autoref{tab: Numerics table}. For the given parameters in \autoref{tab: Numerics table}, the critical value $\Lambda_c$ and the critical wave indices $(m_x^c, m_y^c)$ are found. The critical eigenvalues are real, and the transition number is negative. According to our main theorem, \autoref{main thm}, the system bifurcates from the basic shear flow to  a stable circle of steady states. In \autoref{tab: Numerics table}, we note that some of the coefficients of the center manifold is zero due to the even-odd symmetries of the eigenvectors, canceling some of the integrations.

\begin{table}[tb]
  \caption{Real and imaginary parts of the numerical values for the center manifold and transition number evaluated for the parameter values $L_x = 1$, $L_y=1$, $\Fr=0.2$, $\Pr=0.71$, $\Ro=.1$, $\Lambda = \Lambda_c = 0.799$ and $(m_x^c, m_y^c) = (0, 1)$. The transition number is $A = \sum_{m_z=1}^{6} A_{0,m_z} + A_{2, m_z} = -0.0147 + 0.0021  = -0.0126$.}
  \label{tab: Numerics table}
  \centering
  \pgfplotstabletypeset[col sep=space,
  precision=0,
  every head row/.style={before row=\hline,after row=\hline}, 
  every last row/.style={after row=\hline}, 
  columns/0/.style={column name=$m_z$},
  columns/1/.style={column name=$\Re \Phi_{0,m_z}$},
  columns/2/.style={column name=$\Im \Phi_{0,m_z}$},
  columns/3/.style={column name=$\Re \Phi_{2,m_z}$},
  columns/4/.style={column name=$\Im \Phi_{2,m_z}$},
  columns/5/.style={column name=$\Re A_{0,m_z}$},
  columns/6/.style={column name=$\Im A_{0,m_z}$},
  columns/7/.style={column name=$\Re A_{2,m_z}$},
  columns/8/.style={column name=$\Im A_{2,m_z}$}
  ]{
    1 -1.1551957853570755e-17 1.3426311029812396e-17  -0.029433562364305735 0.0012450936241121657 7.629910466844947e-19 -3.556255648102057e-18  0.0024212113377004786 7.090682208055199e-17
    2 -1.1551957853570756e-17 -1.3426311029812397e-17 -0.000561085989210703 0.00009717387697555772  -7.629910466844964e-19  -3.556255648102061e-18  0.00004819463923795705  1.5246593050577406e-19
    3 4.2995719674040285e-17  0.  0.000013153308953452463 -3.670329838553366e-6 -5.007417855406813e-33  -1.7533717697813653e-18 -6.482940285466235e-7 2.093812506053427e-18
    4 -0.015620408674424192 0.001071509965170519  -0.0049867223978039875  -0.0027109837070407395  -0.011779468146849003 -0.0011365846879320199  -0.00018910627504714427 -0.0001964907311204311
    5 -0.015620408674424194 -0.00107150996517052  0.0053678172753883905 0.0018448226556070696 -0.01177946814684902  0.0011365846879320225 -0.00018910627504711115 0.00019649073112047873
    6 0.01364096589263654 0.  -0.0001241112755316855  0.000472808684054187  0.00884778229557009 -4.336808689942018e-18  -0.000013094569128293718  2.530087413448595e-18
  }
\end{table}

In \autoref{fig: Ro Fr vs A}, the values of transition number $A$ is shown in a parameter regime. In the case of very low $\Ro$ numbers it is seen from the first plot that the transition is to spatiotemporal oscillations. Also observed in this case is that the real part of the transition number admits both negative and positive values indicating both continuous and catastrophic  transitions. As the $\Ro$ number increases, the transition becomes a continuous one to multiple steady states. In the second plot of \autoref{fig: Ro Fr vs A} we test the transition number against the Froude number and find that the transition is continuous and is to multiple steady states.
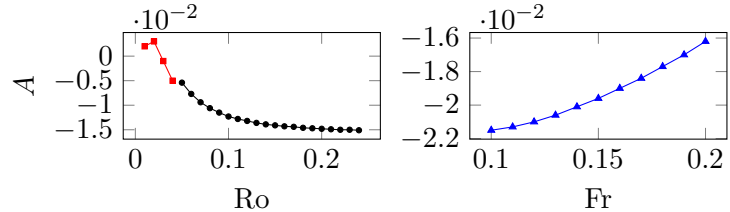
\begin{figure}[H]
  \centering
  \begin{tikzpicture}
  \begin{axis}[
      xlabel=$\Ro$,
      ylabel=$A$,
      height=3cm,
      width=5cm]
  \addplot[red, mark=square*, mark size=1] coordinates { (0.01, 0.002) (0.02, 0.003) (0.03, -0.001) (0.04, -0.005) }; \label{plRovsA1}
  \addplot[mark=*, mark size=1] coordinates { (0.05, -0.0054) (0.06, -0.0077) (0.07, -0.0094) (0.08, -0.0106) (0.09, -0.0115) (0.1, -0.0123) (0.11, -0.0128) (0.12, -0.0132) (0.13, -0.0136) (0.14, -0.0139) (0.15, -0.0141) (0.16, -0.0143) (0.17, -0.0144) (0.18, -0.0146) (0.19, -0.0147) (0.2, -0.0148) (0.21, -0.0149) (0.22, -0.015) (0.23, -0.015) (0.24, -0.0151) }; \label{plRovsA2}
  \end{axis}
  \end{tikzpicture}
  \begin{tikzpicture}
  \begin{axis}[
      xlabel=$\Fr$,
      height=3cm,
      width=5cm]
  \addplot[blue, mark=triangle*, mark size=1.5] coordinates { (0.1, -0.0215) (0.11, -0.0213) (0.12, -0.021) (0.13, -0.0206) (0.14, -0.0201) (0.15, -0.0196) (0.16, -0.019) (0.17, -0.0184) (0.18, -0.0177) (0.19, -0.017) (0.2, -0.0162) }; \label{plFrvsA}
  \end{axis}
  \end{tikzpicture}
  \caption{The transition number $A$ at fixed $L_x=1$, $L_y=1$, $\Fr=0.2$, $\Lambda=\Lambda_c$, $\Pr=0.71$, $Ro=2$ except for the one which is varied. The curve (\ref{plRovsA1}) shows the real part of the transition number where the transition is to spatio-temporal oscillations. In (\ref{plRovsA2}) and (\ref{plFrvsA}), the transition is to multiple steady states.}
  \label{fig: Ro Fr vs A}
\end{figure}

\subsection{Critical Eigenvectors}
\autoref{fig: critical eigenvectors} displays the real and imaginary parts of the $z$-dependence of the velocity and density profiles of the first critical eigenvector for the given parameter values. The wave index of this critical mode is found to be $(m_x^c, m_y^c) = (0, 1)$. This is a roll pattern parallel to the x-axis. The skewed roll structure of the first critical eigenmodes can be seen in \autoref{fig:contourplotdensityvsVW}.
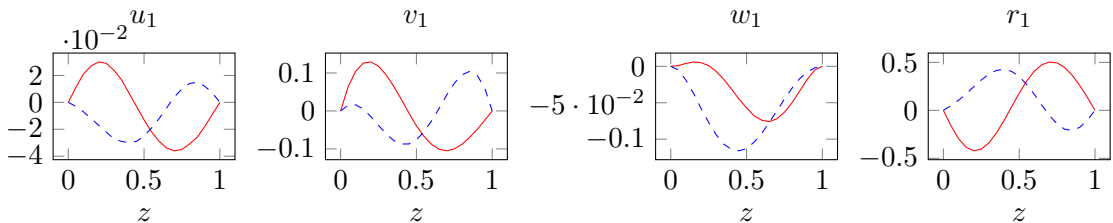
\begin{figure}[H]
  \centering
  \begin{tikzpicture}
  \begin{axis}[
      xlabel=$z$,
      height=3cm,
      width=4cm,
      title=$u_1$]
  \addplot[red] coordinates { (0., 0.) (0.05, 0.011) (0.1, 0.021) (0.15, 0.027) (0.2, 0.03) (0.25, 0.029) (0.3, 0.024) (0.35, 0.017) (0.4, 0.007) (0.45, -0.003) (0.5, -0.013) (0.55, -0.022) (0.6, -0.029) (0.65, -0.034) (0.7, -0.036) (0.75, -0.035) (0.8, -0.031) (0.85, -0.026) (0.9, -0.018) (0.95, -0.009) (1., 0.) }; \label{eigReuc}
  \addplot[blue, dashed] coordinates { (0., 0.) (0.05, -0.003) (0.1, -0.007) (0.15, -0.012) (0.2, -0.017) (0.25, -0.022) (0.3, -0.027) (0.35, -0.029) (0.4, -0.03) (0.45, -0.029) (0.5, -0.025) (0.55, -0.019) (0.6, -0.012) (0.65, -0.004) (0.7, 0.004) (0.75, 0.01) (0.8, 0.014) (0.85, 0.015) (0.9, 0.012) (0.95, 0.007) (1., 0.) }; \label{eigImuc}
  \end{axis}
  \end{tikzpicture}
  \begin{tikzpicture}
  \begin{axis}[
      xlabel=$z$,
      height=3cm,
      width=4cm,
      title=$v_1$]
  \addplot[red] coordinates { (0., 0.) (0.05, 0.064) (0.1, 0.105) (0.15, 0.126) (0.2, 0.129) (0.25, 0.118) (0.3, 0.095) (0.35, 0.065) (0.4, 0.031) (0.45, -0.005) (0.5, -0.039) (0.55, -0.067) (0.6, -0.088) (0.65, -0.101) (0.7, -0.105) (0.75, -0.1) (0.8, -0.088) (0.85, -0.07) (0.9, -0.049) (0.95, -0.025) (1., 0.) }; \label{eigRevc}
  \addplot[blue, dashed] coordinates { (0., 0.) (0.05, 0.015) (0.1, 0.016) (0.15, 0.005) (0.2, -0.015) (0.25, -0.038) (0.3, -0.06) (0.35, -0.078) (0.4, -0.087) (0.45, -0.087) (0.5, -0.076) (0.55, -0.055) (0.6, -0.026) (0.65, 0.007) (0.7, 0.042) (0.75, 0.073) (0.8, 0.095) (0.85, 0.104) (0.9, 0.094) (0.95, 0.061) (1., 0.) }; \label{eigImvc}
  \end{axis}
  \end{tikzpicture}
  \begin{tikzpicture}
  \begin{axis}[
      xlabel=$z$,
      height=3cm,
      width=4cm,
      title=$w_1$]
  \addplot[red] coordinates { (0., 0.) (0.05, 0.001) (0.1, 0.004) (0.15, 0.006) (0.2, 0.005) (0.25, 0.001) (0.3, -0.007) (0.35, -0.018) (0.4, -0.031) (0.45, -0.045) (0.5, -0.057) (0.55, -0.068) (0.6, -0.074) (0.65, -0.076) (0.7, -0.072) (0.75, -0.063) (0.8, -0.05) (0.85, -0.034) (0.9, -0.018) (0.95, -0.005) (1., 0.) }; \label{eigRewc}
  \addplot[blue, dashed] coordinates { (0., 0.) (0.05, -0.005) (0.1, -0.019) (0.15, -0.037) (0.2, -0.058) (0.25, -0.077) (0.3, -0.094) (0.35, -0.107) (0.4, -0.114) (0.45, -0.116) (0.5, -0.113) (0.55, -0.104) (0.6, -0.092) (0.65, -0.077) (0.7, -0.061) (0.75, -0.045) (0.8, -0.03) (0.85, -0.017) (0.9, -0.008) (0.95, -0.002) (1., 0.) }; \label{eigImwc}
  \end{axis}
  \end{tikzpicture}
  \begin{tikzpicture}
  \begin{axis}[
      xlabel=$z$,
      height=3cm,
      width=4cm,
      title=$r_1$]
  \addplot[red] coordinates { (0., 0.) (0.05, -0.161) (0.1, -0.295) (0.15, -0.385) (0.2, -0.422) (0.25, -0.405) (0.3, -0.34) (0.35, -0.236) (0.4, -0.106) (0.45, 0.037) (0.5, 0.179) (0.55, 0.306) (0.6, 0.408) (0.65, 0.475) (0.7, 0.504) (0.75, 0.494) (0.8, 0.446) (0.85, 0.365) (0.9, 0.258) (0.95, 0.134) (1., 0.) }; \label{eigRerc}
  \addplot[blue, dashed] coordinates { (0., 0.) (0.05, 0.048) (0.1, 0.105) (0.15, 0.172) (0.2, 0.245) (0.25, 0.316) (0.3, 0.376) (0.35, 0.416) (0.4, 0.427) (0.45, 0.406) (0.5, 0.352) (0.55, 0.269) (0.6, 0.166) (0.65, 0.053) (0.7, -0.054) (0.75, -0.142) (0.8, -0.197) (0.85, -0.21) (0.9, -0.177) (0.95, -0.102) (1., 0.) }; \label{eigImrc}
  \end{axis}
  \end{tikzpicture}
  \caption{The real part (\ref{eigRerc}) and the imaginary part (\ref{eigImrc}) of the z dependence of the density and the velocity profiles of the first critical eigenvector at $L_x = L_y =1$, $\Fr = 0.2$, $\Pr = 0.71$, $\Ro = 0.1$, $\Lambda = \Lambda_c$. Here $(m_x^c, m_y^c) = (0, 1)$.}
  \label{fig: critical eigenvectors}
\end{figure}

\begin{figure}[tb]
  \centering
  \includegraphics[scale=1]{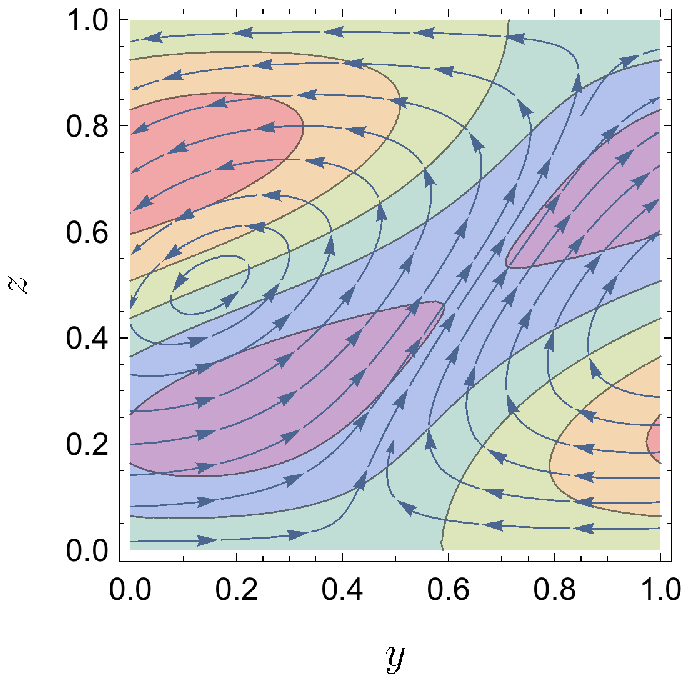}
  \includegraphics[scale=1]{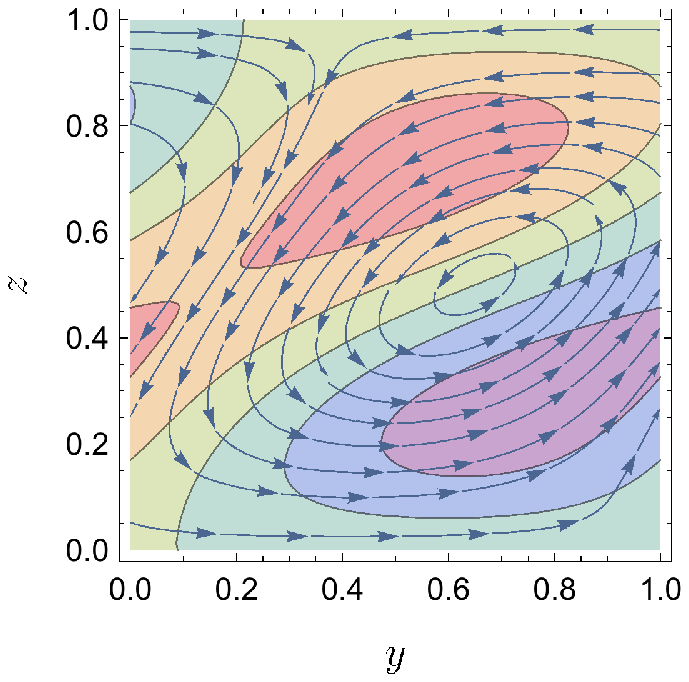}
  \caption{The real and imaginary parts of $(v_1, w_1)$-vector field on top of the contour plot of the critical density field in the $yz$-plane for $L_x = L_y =1$, $\Fr = 0.2$, $\Pr = 0.71$, $\Ro = 0.1$, $\Lambda = \Lambda_c$. Here $(m_x^c, m_y^c) = (0, 1)$ .}
  \label{fig:contourplotdensityvsVW}
\end{figure}

\subsection{Critical Wave Number Selection and Pattern Formation}
In this section, we give some remarks on the selection of the critical wave number and its implications on the pattern formation.

We fix the parameters as $\Fr=0.2$, $\Pr=0.71$ and $\Ro=0.1$ and obtain \autoref{fig:Wave Number Selection} describing the selection of wave indices for the critical eigenmodes.

From \autoref{fig:Wave Number Selection} at the indicated parameters, we observe that regardless of the horizontal length scale $L_x$, the wave indices are given by $(0,1)$ if $L_y < 1.74$, $(0,2)$ if $1.74<L_y<2.96$, $(0,3)$ if $2.96<L_y<4.17$. We observed the similar behavior for the selection of the critical wave indices in our other numerical experiments.

Thus the first critical modes always consist of $m_y$-rolls aligned with the x-axis irrespective of $L_x$. Moreover $m_y$, the number of rolls increases as $L_y$ is increased.

It is interesting that the system does not seem to admit rectangle modes (i.e. $(m_x, m_y)$ with $m_x\ne0$ and $m_y \ne 0)$) or rolls parallel to the y-axis.

\begin{figure}[H]
  \centering
  \begin{tikzpicture}
    \begin{axis}[scatter/classes={
      m01={mark=*, blue},%
      m02={mark=square*, red},%
      m03={mark=triangle*, orange} },
      width=7cm,
      height=3.5cm,
      xlabel=$L_x$,
      ylabel=$L_y$,
      legend pos=outer north east
      ]
      \addplot+[only marks, mark size=1.5, mark options={solid}, scatter, scatter src=explicit symbolic] table[x=Lx,y=Ly, meta=mxmy] {
        Lx  Ly  mxmy
        0.1 0.5 m01
        0.5 0.5 m01
        1.  0.5 m01
        2.  0.5 m01
        3.  0.5 m01
        0.1 1.  m01
        0.5 1.  m01
        1.  1.  m01
        2.  1.  m01
        3.  1.  m01
        0.1 1.7 m01
        0.5 1.7 m01
        1.  1.7 m01
        2.  1.7 m01
        3.  1.7 m01
        0.1 1.9 m02
        0.5 1.9 m02
        1.  1.9 m02
        2.  1.9 m02
        3.  1.9 m02
        0.1 3.  m02
        0.5 3.  m02
        1.  3.  m02
        2.  3.  m02
        3.  3.  m02
        0.1 3.2 m03
        0.5 3.2 m03
        1.  3.2 m03
        2.  3.2 m03
        3.  3.2 m03
      };
      \legend{ (0,1), (0,2), (0,3)}
    \end{axis}
  \end{tikzpicture}
  \caption{The selection of the critical wave number at $\Fr=0.2$, $\Pr=0.71$ and $\Ro=0.1$.}
  \label{fig:Wave Number Selection}
\end{figure}
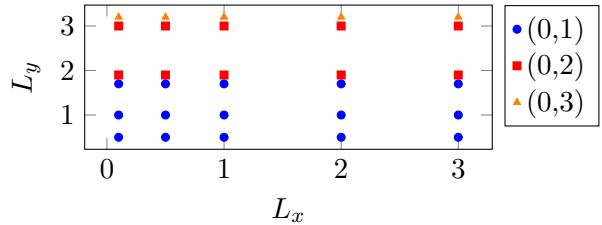

\section{Conclusions}
We find a criterion for the energy stability of the system in terms of system parameters.

By studying the linear stability of the problem, we find a criterion for the local nonlinear stability of the system. Our numerical investigations show the existence of a hypersurface $\Lambda = \Lambda_c$ which separates the parameter space into regions where the basic shear flow is stable and unstable. More explicit   expression for $\Lambda_c$ is not given in this article, and will be addressed  in a later study.

Our numerical investigations also yield that selection of horizontal wave indices at fixed $\Fr$, $\Pr$, $\Ro$ is determined by only $L_y$ and is independent of $L_x$. We find that the system admits only critical eigenmodes with horizontal wave indices $(0,m_y)$. Such modes, horizontally have the pattern consisting of $m_y$-rolls aligned with the x-axis.

We show under the assumption of critical crossing of eigenvalues that the system exhibits dynamic transitions to both multiple steady states and spatiotemporal oscillations. In both cases, the transition can be either continuous, represented by local attractors close to the basic shear flow, or catastrophic, represented  by local attractors far away from the basic shear flow. Numerically we encountered continuous transitions to multiple steady states, continuous and catastrophic transitions to spatiotemporal oscillations.

Another contribution of  this paper is the derivation of a representation of the center manifold approximation in the case of a pair of complex conjugate eigenvalues, adapting similar results \cite{ptd} to the present case. This new representation allows to write a unified transition number covering both the double real eigenvalue and a pair of complex conjugate eigenvalue cases.

\appendix
\section{Approximation of Center Manifold Functions}\label{s11}
In this section we derive a new representation for the approximation of the center manifold function suitable for studying problems in spatial domains with periodicity condition in at least one direction. General approximate formulas for center manifold functions  were first introduced \cite{b-book, MW05c, ptd}, and have played central role in a  wide range of applications of the dynamic transition theory.

We follow the same framework  as in \cite{ptd}. Let $X$ and $X_{1}$ be two Banach spaces and $X_{1}\subset X$ a compact and dense inclusion. Consider the following evolution equation
\begin{equation} \label{CH0:main}
\begin{aligned}
& \frac{d \phi}{dt} =L_{\lambda }\phi+G\left( \phi,\lambda \right) ,\\
& \phi\left( 0\right) =\varphi,
\end{aligned}
\end{equation}
where $\phi$ is the unknown function and $\lambda \in \mathbb{R}$ is the
parameter.

We assume that $L_{\lambda }=-A+B_{\lambda }:X_{1}\rightarrow X$ is a linear completely
continuous field, where $A:X_{1}\rightarrow X$ is a linear homeomorphism and $B_{\lambda
}:X_{1}\rightarrow X$ is a linear compact operator. Furthermore we assume
that $L_{\lambda }$ is a sectorial operator depending continuously on $
\lambda $. In this case, we can define the fractional order spaces $
X_{\alpha}=D\left( L_{\lambda }^{\alpha }\right) $ for $\alpha \in \mathbb{R}$.
We also assume that $G\left( \cdot ,\lambda \right) :X_{\theta }\rightarrow
X $ is $C^{r}$ $\left( r\geq 1\right) $ bounded mapping for some $0\leq
\theta <1$, depending continuously on $\lambda \in \mathbb{R}$ and
\begin{equation}\label{CH0:Theta in CMT}
G\left( \phi,\lambda \right) =o\left( \norm{\phi}_{X_{\theta}}\right) \qquad \forall \lambda \in \mathbb{R}.
\end{equation}

We know that for a linear completely continuous field, the spectrum consists of isolated eigenvalues with finite dimensional eigenspaces.
Let $\{\, \beta_i(\lambda) \in \mathbb{C} \:|\: i \in \mathbb{N}\}$ be all eigenvalues of $L_\lambda$ counting multiplicities and $\{\phi_i(\lambda)\:|\: i \in \mathbb{N}\}$ be the corresponding (complex) eigenvectors.
Assume that the following principle of exchange of stabilities (PES) condition holds true:
\begin{equation}
\label{CH0:PES}
\begin{aligned}
& \Re\beta_1 = \Re \beta_2
\begin{cases}
  <0 & \text{ if }  \lambda <\lambda_{0}, \\
  =0 & \text{ if } \lambda =\lambda_{0},  \\
  >0 & \text{ if } \lambda >\lambda_{0},
\end{cases} \\
& \Re\beta _{j}\left( \lambda _{0}\right) <0 && \text{ for }j\geq 3.
\end{aligned}
\end{equation}
Assume the eigenpairs satisfy
\[
  \beta_1 = \overline{\beta_2}, \qquad \phi_1 = \overline{\phi_2}.
\]
By the spectral theorem \cite{b-book, ptd}, the spaces $X_1$ and $X$ can be decomposed into the direct sum
\begin{equation}\label{E1E2}
  X_1=E_1^{\lambda} \oplus E_2^{\lambda}, \quad X = E_1^{\lambda} \oplus \overline{E_2^{\lambda}},
\end{equation}
where
\begin{equation*}
\begin{aligned}
E_1^\lambda &=\text{span}\{z \phi_1(\lambda) + \overline{z \phi_1(\lambda)} \mid z \in \mathbb{C}\},\\
E_2^\lambda &= \text{the complement of } E_1^\lambda \text{ in } X.
\end{aligned}
\end{equation*}
Then $L_\lambda$ is invariant on $E_1^\lambda$ and $E_2^\lambda$, i.e., $L_\lambda$ can be decomposed as
\begin{equation} \label{decomp L}
\begin{aligned}
&L_\lambda = \mathcal{J}_\lambda \oplus \mathcal{L}_{\lambda}, \\
&\mathcal{J}_\lambda: E_1^\lambda \rightarrow E_1^\lambda, \\
&\mathcal{L}_{\lambda}: X_1\cap E_2^\lambda
\rightarrow E_2^\lambda.
\end{aligned}
\end{equation}
Finally let $P_i$ be the projection onto $E_i^{\lambda}$. Below we omit $\lambda$'s to simplify the notation.
\begin{lemma}
Assume $G$ defined in \eqref{CH0:main} is a bilinear operator. Let
\[
  x(t) = z(t) \phi_1 + \overline{z(t) \phi_1} \in E_1.
\]
Under the above assumptions, for $\lambda$ sufficiently close to $\lambda_0$ we have the following approximation for
the center manifold function
\begin{equation} \label{new cm formula}
  \Phi(x(t), \lambda) = (2 \beta_1 - \mathcal{L}_{\lambda})^{-1} P_2G (z \phi_1, z\phi_1) + (-\mathcal{L}_{\lambda})^{-1} P_2G(z \phi_1, \overline{z \phi_1}) + {\rm o}(2) + \text{ c.c.},
\end{equation}
where $\text{c.c.}$ stands for the complex conjugate of the whole expression coming before and ${\rm o}(k)$ stands for
\begin{equation} \label{little o}
{\rm o}(k):={\rm o}(\norm{x}_{X_1}^k)+{\rm O}( \abs{\Re \beta_1(\lambda)} \norm{x}_{X_1}^k).
\end{equation}
\end{lemma}
\begin{proof}
It is known that the center manifold can be approximated by
\begin{equation} \label{comp proof 1}
  \Phi = \int_{-\infty}^0 e^{-\tau \mathcal{L}_{\lambda}} \rho_{\epsilon} P_2 G(e^{\tau \mathcal{J}_{\lambda} } x) d\tau + {\rm o}(k),
\end{equation}
where $\mathcal{J}_{\lambda} = \text{diag}(\beta_1, \overline{\beta_1})$ so that $e^{\tau \mathcal{J}_{\lambda} } x = e^{\tau \beta_1} z \phi_1 + \overline{ e^{\tau \beta_1} z \phi_1}$ with $\rho_{\epsilon}$ denoting a cut-off function. Hence
\[
  G(e^{\tau \mathcal{J}_{\lambda} } x) = e^{2 \tau \beta_1} G(z \phi_1, z \phi_1) + e^{\tau (\beta_1 + \overline{\beta_1})} G(z \phi_1, \overline{z \phi_1}) + \text{c.c.}
\]
Let
\begin{equation}\label{gmhm}
  P_2 G(z \phi_1, z \phi_1) = \sum_{m \ge 3} g_m \phi_m, \qquad
  P_2 G(z \phi_1, \overline{z \phi_1}) = \sum_{m \ge 3} h_m \phi_m.
\end{equation}
Then
\begin{align} \label{comp proof 2}
 \Phi & =  \int_{-\infty}^0 e^{-\tau \mathcal{L}_{\lambda}} \sum_{m \ge 3}(e^{2\tau \beta_1} g_m + e^{\tau (\beta_1 + \bar \beta_1)} h_m)\phi_m d\tau + \text{c.c.} + {\rm o}(2) \\
 &  =  \sum_{m \ge 3} \int_{-\infty}^0 (e^{-\tau (\beta_m - 2\beta_1)} g_m + e^{-\tau (\beta_m - \beta_1 - \bar \beta_1)} h_m)\phi_m d\tau + \text{c.c.} + {\rm o}(2). \nonumber
\end{align}

The following integrals converge by the PES condition, since $\Re(\beta_m - 2\beta_1) = \Re(\beta_m - \beta_1 - \bar \beta_1)< 0$:
\begin{align*}
 &
  \int_{-\infty}^0 e^{-\tau (\beta_m - 2\beta_1)} d \tau = \frac{1}{2\beta_1 - \beta_m}, \\
  &
  \int_{-\infty}^0 e^{-\tau (\beta_m - \beta_1 -\bar \beta_1)} d \tau = \frac{1}{2\Re(\beta_1) - \beta_m} = -\frac{1}{\beta_m} + O(\abs{\Re \beta_1 }).
\end{align*}
Plugging these expressions in \eqref{comp proof 2}, we obtain
\begin{align*}
\Phi = & \sum_{m \ge 3} \left( \frac{1}{2 \beta_1 - \beta_m} g_m - \frac{1}{\beta_m} h_m \right) \phi_m + \text{c.c.} + o(2) \\
= & (2 \beta_1 - \mathcal{L}_{\lambda})^{-1} \sum_{m \ge 3} g_m \phi_m + (- \mathcal{L}_{\lambda})^{-1} \sum_{m \ge 3} h_m \phi_m + \text{c.c.} + o(2),
\end{align*}

Finally using \eqref{gmhm} gives the desired result.
\end{proof}

\begin{remark}
  An alternative way of representing the center manifold approximation given in \cite{ptd} under the above assumptions when $\Im\beta_1 = -\Im \beta_2 = \rho \ne 0$ and $G = G_2$ is bilinear is
  \begin{equation} \label{old cm complex formula}
    \begin{split}
      & ( ( -\mathcal{L}_{\lambda})^2 + 4\rho ^2) ( -\mathcal{L}_{\lambda}) \Phi(x_1, x_2, \lambda)  = \\
      & ( ( -\mathcal{L}_{\lambda}) ^2 + 2\rho ^2) P_2G_2( x_1 \Re(\phi_1) + x_2 \Im(\phi_1)) \\
      & + 2\rho ^{2}P_{2}G_2(x_1\Im(\phi_1)-x_2\Re \phi_1)\\
      & + \rho (-\mathcal{L}_{\lambda}) P_2 G_2( x_1\Re \phi_1+x_2\Im(\phi_1),x_2\Re(\phi_1)-x_1\Im(\phi_1))  \\
      & + \rho (-\mathcal{L}_{\lambda}) P_2 G_2( x_2\Re(\phi_1)-x_1\Im(\phi_1),x_1\Re(\phi_1)+x_2\Im(\phi_1)) +o( 2).
    \end{split}
  \end{equation}
  It should be noted that the new representation \eqref{new cm formula} is less prone to mistakes in center manifold approximation computations than the representation \eqref{old cm complex formula}.
\end{remark}

\begin{remark}
  We note that if $\Im \beta_1 = 0$, the formula \eqref{new cm formula} reduces to the case
\begin{equation} \label{real cm formula}
  -\mathcal{L}_{\lambda} \Phi(x(t), \lambda) = P_2 G_2(x(t)) + o(2),
\end{equation}
which was first derived in \cite{b-book, MW05c}. Thus the formula \eqref{new cm formula} gives a unified approach to the center manifold approximation formulas in the cases of both real and non-real crossing of two critical eigenvalues from the imaginary axis at the critical parameter.
\end{remark}

\begin{remark}
  The formula \eqref{new cm formula} can be expressed in a more compact form as follows
  \begin{equation} \label{new cm formula compact form}
    \Phi(x(t), \lambda) = \sum_{i_1, i_2 = 1}^2 (\beta_{i_1} + \beta_{i_2} - \mathcal{L}_{\lambda})^{-1} P_2 G_2(u_{i_1}, u_{i_2}) + o(2)
  \end{equation}
  where $u_1 = \overline{u_2} = z(t) \phi_1$.

  The above form of the center manifold approximation can be extended to the multi-linear case. That is if $G = G_n$ is an n-linear operator then
  \begin{equation} \label{new cm formula Gn case}
    \Phi(x(t), \lambda) = \sum_{i_1, \dots, i_n = 1}^2 (\beta_{i_1} + \cdots + \beta_{i_n} - \mathcal{L}_{\lambda})^{-1} P_2 G_n(u_{i_1}, \dots, u_{i_n}) + o(n),
  \end{equation}
  where $u_1 = \overline{u_2} = z(t) \phi_1$.


\end{remark}

\bibliographystyle{siam}
\def\cprime{$'$}

\end{document}